\documentclass[journal]{IEEEtran}

\usepackage{times}
\usepackage{bm,bbm}
\usepackage{amsmath,amssymb}
\usepackage{graphicx,subfigure}
\usepackage{url}
\usepackage{units}
\usepackage{cite,balance}
\usepackage{dblfloatfix}
\usepackage{multirow}
\usepackage{booktabs}

\newtheorem{definition}{Definition}
\newtheorem{proposition}{Proposition}
\newcommand{\at}[2][]{#1|_{#2}}

\begin{document}

\title{Parameter Estimation in SAR Imagery using Stochastic Distances and Asymmetric Kernels}

\author{ 
Juliana~Gambini, Julia~Cassetti, Mar\'{\i}a Magdalena Lucini, and Alejandro~C.~Frery,~\IEEEmembership{Senior Member}

\thanks{This work was supported by Conicet, CNPq, and Fapeal.}
\thanks{Juliana Gambini is with the Instituto Tecnol\'ogico de Buenos Aires, Av.\ Madero 399, C1106ACD Buenos Aires,  Argentina and with Depto. de Ingenier\'{\i}a en Computaci\'on, Universidad Nacional de Tres de Febrero, Pcia.\ de Buenos Aires, Argentina, \texttt{juliana.gambini@gmail.com}.}
\thanks{Julia Cassetti is with the  Instituto de Desarrollo Humano, Universidad Nacional de Gral.\ Sarmiento, Pcia.\ de Buenos Aires, Argentina.}
\thanks{Magdalena Lucini is with the Facultad de Ciencias Exactas, Naturales y Agrimensura, Universidad Nacional del Nordeste, Av. Libertad 5460, 3400 Corrientes, Argentina.}
\thanks{Alejandro C.\ Frery is with the LaCCAN, Universidade Federal de Alagoas, 
Av. Lourival Melo Mota, s/n, 57072-900 Macei\'o -- AL, Brazil, 
\texttt{acfrery@gmail.com}}
}

\maketitle

\begin{abstract}
In this paper we analyze several strategies for the estimation of the roughness parameter of the $\mathcal G_I^0$ distribution.
It has been shown that this distribution is able to characterize a large number of targets in monopolarized SAR imagery, deserving the denomination of ``Universal Model.''
It is indexed by three parameters: the number of looks (which can be estimated in the whole image), a scale parameter, and the roughness or texture parameter.
The latter is closely related to the number of elementary backscatters in each pixel, one of the reasons for receiving  attention in the literature.
Although there are efforts in providing improved and robust estimates for such quantity, its dependable estimation still poses numerical problems in practice.
We discuss estimators based on the minimization of stochastic distances between empirical and theoretical densities, and argue in favor of using an estimator based on the Triangular distance and asymmetric kernels built with Inverse Gaussian densities.
We also provide new results regarding the heavytailedness of this distribution.
\end{abstract}


\IEEEpeerreviewmaketitle

\section{Introduction}
\PARstart{T}{he} statistical modeling of the data is essential in order to interpret SAR images.
Speckled data have been described under the multiplicative model using the $\mathcal{G}$ family of distributions which is able to describe rough and extremely rough areas better than the $\mathcal{K}$ distribution~\cite{Frery97,MejailJacoboFreryBustos:IJRS}. 
The survey article~\cite{Gao2010} discusses in detail several statistical models for this kind of data. 

Under the $\mathcal{G}$ model different degrees of roughness are associated to different parameter values, therefore it is of paramount importance to have high quality estimators.
Several works have been devoted to the subject of improving estimation with two main venues of research, namely, analytic and resampling procedures.

The analytic approach was taken by Vasconcellos et al.~\cite{VasconcellosFrerySilva:CompStat} who quantified the bias in the estimation of the roughness parameter of the $\mathcal G_A^0$ distribution by maximum likelihood (ML). 
They proposed an analytic change for improved performance with respect to bias and mean squared error.
Also, Silva et al.~\cite{SilvaCribariFrery:ImprovedLikelihood:Environmetrics} computed analytic improvements for that estimator.
Such approaches reduce both the bias and the mean squared error of the estimation, at the expense of computing somewhat cumbersome correction terms, and yielding estimators whose robustness is largely unknown.

Cribari-Neto et al.~\cite{CribariFrerySilva:CSDA} compared several numerical improvements for that estimator using bootstrap.
Again, the improvement comes at the expense of intensive computing and with no known properties under contamination.
Allende et al.~\cite{AllendeFreryetal:JSCS:05} and Bustos et al.~\cite{BustosFreryLucini:Mestimators:2001} also sought for improved estimation, but seeking the robustness of the procedure.
The new estimators are resistant to contamination, and in some cases they also improve the mean squared error, but they require dealing with influence functions and asymptotic properties not always immediately available to remote sensing practitioners.

A common issue in all the aforementioned estimation procedures, including ML, and to those based on fractional moments~\cite{Frery97} and log-cumulants~\cite{MellinAnalysisPolSAR,BujorTrouveValetNicolas2004,khan2014} is the need of iterative algorithms for which there is no granted convergence to global solutions.
Such lack of convergence usually arises with small samples, precluding the use of such techniques in, e.g., statistical filters.
Frery et al.~\cite{FreryCribariSouza:JASP:04} and Pianto and Cribari-Neto~\cite{DealingMonotoneLikelihood} proposed techniques which aim at alleviating such issue, at the cost of additional computational load.

The main idea of this work is to develop an estimation method for the $\mathcal G_I^0$ 	model with good properties (as measured by its bias, the mean squared error and its ability to resist contamination) even with samples of small and moderate size, and low computational cost.
In order to achieve this task, we propose minimizing a stochastic distance between the fixed empirical evidence and the estimated model.

Shannon proposed a divergence between two density functions as a measure of the relative information between the two distributions. 
Divergences were studied by Kullback and Leibler and by R\'enyi~\cite{arndt2004information}, among others.
These divergences have multiple applications in signal and image processing~\cite{4218961}, medical image analysis diagnosis~\cite{5599869}, and automatic region detection in SAR imagery~\cite{5208318,ClassificationPolSARSegmentsMinimizationWishartDistances,EdgeDetectionDistancesEntropiesJSTARS,SARSegmentationLevelSetGA0}.
Liese and Vajda~\cite{Liese2006} provide a detailed theoretical analysis of divergence measures.

Cassetti et al.~\cite{APSAR2013ParameterEstimationStochasticDistances} compared estimators based on the Hellinger, Bhattacharyya, R\'enyi and Triangular distance with the ML estimator.
They presented evidence that the Triangular distance is the best choice for this application, but noticed that histograms led to many numerical instabilities.
This work presents improvements with respect to those results in the following regards: we assess the impact of contamination in the estimation,  we employ kernels rather than histograms and we compare estimators based on the Triangular distance with the ML, Fractional Moments and Log-Cumulants estimators.
Among the possibilities for such estimate, we opted for employing asymmetric kernels.
 
Kernels have been extensively applied with success to image processing problems as, for instance, 
object tracking~\cite{Bohyung2008}.
Kernels with positive support, which are of particular interest for our work, were employed in~\cite{chensx2000,scaillet2004}.
The general problem of using asymmetric kernels was studied in~\cite{ECT291329}. 

In~\cite{Achim03sarimage}, the authors demonstrate, through extensive modeling of real data, that SAR images are best described by families of heavy-tailed distributions. 
We provide new results regarding the heavytailedness of the $\mathcal G_I ^0$ distribution; in particular, we show that it has heavy tails with tail index $1-\alpha$.

The paper unfolds as follows. 
Section~\ref{sec_SAR} recalls the main properties of the $\mathcal{G}_I^0$ model and  the Maximum Likelihood, the $\frac{1}{2}$-moment and the log-cumulants methods for parameter estimation. 
Estimation with stochastic distances is presented in Section~\ref{estimation}, including different ways of calculating the  estimate of the underlying density function and the contamination models employed to assess the robustness of the procedure.
Section~\ref{results} presents the main results.
Section~\ref{conclusions} discusses the conclusions. 

\section{The $\mathcal{G}_I^0$ Model}\label{sec_SAR}

The return in monopolarized SAR images can be modeled as the product of two independent random variables, one corresponding to the backscatter $X$ and the other to the speckle noise $Y$. 
In this manner $Z=X Y  $ represents the return in each pixel under the multiplicative model.

For monopolarized data, speckle is modeled as a $\Gamma $ distributed random variable, with unitary mean and shape parameter $L\geq1$, the number of looks, while the backscatter is considered to obey  a reciprocal of Gamma law. 
This gives rise to the $\mathcal{G}_I^{0}$ distribution for the return.
Given the mathematical tractability and descriptive power of the $\mathcal{G}_I^{0}$ distribution for intensity data~\cite{MejailJacoboFreryBustos:IJRS,mejailfreryjacobobustos2001} it represents an attractive choice for SAR data modeling.

The density function for intensity data is given by
\begin{equation}
f_{\mathcal{G}_I^{0}}( z) =\frac{L^{L}\Gamma ( L-\alpha
) }{\gamma ^{\alpha }\Gamma ( -\alpha ) \Gamma (
L) }  
\frac{z^{L-1}}{( \gamma +zL) ^{L-\alpha }},%
\label{ec_dens_gI0}
\end{equation}
where $-\alpha,\gamma ,z>0$ and $L\geq 1$. 
The $r$-order moments are
\begin{equation}
E(Z^r) =\Big(\frac{\gamma}{L}\Big)^r\frac{\Gamma ( -\alpha-r )}{ \Gamma (-\alpha) }
\frac{\Gamma (L+r )}{\Gamma (L)},
\label{moments_gI0}
\end{equation}
provided $\alpha<-r$, and infinite otherwise.

With the double purpose of simplifying the calculations and making the results comparable, in the following we choose the scale parameter such that $E(Z)=1$, which is given by $\gamma^* =-\alpha-1$. 

One of the most important features of the $\mathcal{G}_I^0$ distribution is the interpretation of the $\alpha$  parameter, which is related to the roughness of the target. 
Values close to zero (typically above $-3$) suggest extreme textured targets, as urban zones. 
As the value decreases, it indicates regions with moderate texture (usually $\alpha \in [-6,-3]$), as forest zones.
Textureless targets, e.g. pasture, usually produce $\alpha\in(-\infty,-6)$. 
This is the reason why the accuracy in the estimation of $\alpha$ is so important.

Let $\bm{z}=(z_1,\dots, z_n)$ be a random sample of $n$ independent draws from the $\mathcal G_I ^0$ model.

Assuming $\gamma^*=-\alpha-1$, the ML estimator for the parameter $\alpha$, namely $\widehat\alpha_{\text{ML}}$ is the solution of the following nonlinear equation:
\begin{align*}
\Psi^0(\widehat{\alpha}_{\text{ML}})-\Psi^0(L-\widehat{\alpha}_{\text{ML}})-\log(1-\widehat{\alpha}_{\text{ML}})+{}\\
\frac{\widehat{\alpha}_{\text{ML}}}{1-\widehat{\alpha}_{\text{ML}}}+
\frac{1}{n}\sum_{i=1}^n{\log(1-\widehat{\alpha}_{\text{ML}}+Lz_i)}-{}
\\ \frac{\widehat{\alpha}_{\text{ML}}-L}{n}\sum_{i=1}^n \frac1{1-\widehat{\alpha}_{\text{ML}}+Lz_i}= 0 ,
\label{derloglikelihood_gI0}
\end{align*}
where $\Psi^0(\cdot)$ is the digamma function.
Some of the issues posed by the solution of this equation have been discussed, and partly solved, in~\cite{FreryCribariSouza:JASP:04}.

Fractional moments estimators have been widely used with success~\cite{Frery97,GambiniSC08}.
Using $r=1/2$ in~\eqref{moments_gI0} one has to solve
\begin{equation}
\frac{1}{n} \sum_{i=1}^n \sqrt{z_i} -\sqrt{\frac{-\widehat\alpha_{\text{Mom12}}-1}{L}}\frac{\Gamma ( -\widehat\alpha_{\text{Mom12}}-{\frac{1}{2}} )}{ \Gamma (-\widehat\alpha_{\text{Mom12}}) }
\frac{\Gamma (L+{\frac{1}{2}} )}{\Gamma (L)}=0.
\label{estim_moment1_2_gI0}
\end{equation}

Estimation based on log-cumulants is gaining space in the literature due to its nice properties and good performance~\cite{MellinAnalysisPolSAR,BujorTrouveValetNicolas2004,khan2014}. Following Tison et al.~\cite{Tison2004} the main second kind statistics can be defined as:
\begin{itemize}
\item First second kind characteristic function:  $\phi_x(s) = \int_0^{\infty} u^{s-1}p_x(u)du$
\item Second  second kind characteristic function: $\psi_x(s) = \log\phi_x(s)$
\item First order second kind characteristic moment:
\begin{equation}
 \tilde{m}_r = \frac{d \phi_x(s)}{ds } \at[\Big]{s=1}. \label{eq:fcm}
\end{equation}
 \item First order second kind characteristic cumulant (log-cumulant)
\begin{equation}
 \tilde{k}_r = \frac{d \psi_x(s)}{ds} \at[\Big]{s=1}. \label{eq:scm}
\end{equation}
\end{itemize}

If $p_x(u) = f_{\mathcal{G}_I^0}(u)$ then we have:
\begin{itemize}
  \item $\phi_x(s) = \frac{\left(\frac{L}{\gamma}\right)^{1-s}\Gamma(-1+L+s)\Gamma(1-s-\alpha)}{\Gamma(L)\Gamma(-\alpha)}$

  \item $\widetilde{m}_1 =  \frac{d \phi_x(s)}{ds } \at[\big]{s=1} = -\log(\frac{L}{\gamma}) + \Psi^0(L) - \Psi^0(-\alpha)$.
\end{itemize}

Using the developments presented in~\cite{Tison2004}, we have that $\widetilde{k}_1 = \widetilde{m_1}$ and the empirical expression for the first log-cumulant estimator for $n$ samples $z_i$ is $\widehat{\widetilde{k}}_1 ={n}^{-1} \sum_{i=1}^n\log z_i$. Therefore
\begin{equation}
  \widetilde{k}_1 =   -\log \frac{L}{\gamma} + \Psi^0(L) - \Psi^0(-\alpha). \label{eq:scGi}
\end{equation}

Assuming $\gamma^*=-\alpha-1$, the Log-Cumulant estimator of $\alpha$, denoted by $\widehat\alpha_{\text{LCum}}$ is then the solution of 
$\widehat{\widetilde{k}}_1 =   -\log \frac{L}{1-\widehat\alpha_{\text{LCum}}} + \Psi^0(L) - \Psi^0(-\widehat\alpha_{\text{LCum}})$, that is, the solution of 
\begin{equation} \label{eq:logm}
  \frac{1}{n} \sum_{i=1}^n\log z_i =   -\log \frac{L}{1-\widehat\alpha_{\text{LCum}}} + \Psi^0(L) + \Psi^0(-\widehat\alpha_{\text{LCum}}).
\end{equation}

The ability of these estimators to resist outliers has not still been assessed.

In the following we provide new results regarding the heavytailedness of the $\mathcal G_I^0$ distribution.
This partly explains the numerical issues faced when seeking for estimators for its parameters: the distribution is prone to producing extreme values.
The main concepts are from~\cite{Gre,Jorgensen,Rojo}.

\begin{definition} \label{Def:lenta}
The function $\ell\colon\mathbb R \to\mathbb R$ is slow-varying in the infinite if for every $t>0$ holds that
$$
\lim_{x\to+\infty}\dfrac{\ell (tx)}{\ell(x)}=1.
$$
\end{definition}

\begin{definition} 
A probability density function $f(x)$ has heavy tails if for any $\eta >0$ holds that
$$
f(x)=\ell(x)  x^{-\eta},
$$
where $\ell$ is a slow-varying function in the infinite, and $\eta$ is the tail index.
\end{definition}
The smaller the tail index is, the more prone to producing extreme observations the distribution is.

\begin{proposition}
The $\mathcal G_{I}^0$ distribution has heavy tails with tail index $1-\alpha$.
\end{proposition}
\begin{proof}
Defining $\ell(x)=f_{\mathcal G_I^0}(x) x^{-\alpha+1}$ we have that
\begin{align*}
\lim_{x\to+\infty}\frac{\ell(t x)}{\ell(x)}&=\lim_{x\to+\infty}\frac{(tx)^{L-1} \, (Ltx+\gamma)^{\alpha-L} \, (tx)^{1-\alpha}}{x^{L-1} \, (Lx+\gamma)^{\alpha-L} \, x^{1-\alpha}}\\
&=\lim_{x\to+\infty}t^{L-\alpha}\left(\frac{Lx+\gamma}{L tx+\gamma}\right)^{L-\alpha}=1.
\end{align*}
This holds for every $t>0$, so $\ell$ is as in Def.~\ref{Def:lenta}
\end{proof}
As expected, the tail index is a decreasing function on $\alpha$, then the $\mathcal G_I^0$ distribution is more prone to producing extreme observations when the roughness parameter is bigger.

The remainder of this section is devoted to proving that the $\mathcal G_I^0$ model is outlier-prone.
Consider the random variables $z_1,\dots,z_n$ and the corresponding order statistics $Z_{1:n}\leq \cdots \leq Z_{n:n}$.
\begin{definition}
The distribution is absolutely outlier-prone if there are positive constants $\varepsilon,\delta$ and an integer $n_0$ such that
$$
\Pr(Z_{n:n}-Z_{n-1:n}> \varepsilon) \geq \delta
$$
holds for every integer $n\geq n_0$.
\end{definition}
A sufficient condition for being absolutely outlier-prone is that there are positive constants $\varepsilon,\delta,x_0$ such that the density of the distribution satisfies, for every $x\geq x_0$, that
\begin{equation}
\frac{f(x+\varepsilon)}{f(x)} \geq \delta. \label{eq:RelAbsOutProne}
\end{equation}
Since
\begin{align*}
\lim_{x\to+\infty}\frac{f_{\mathcal G_I^0}(x+\varepsilon)}{f_{\mathcal G_I^0}(x)} & =
\lim_{x\to+\infty}\Big(\frac{x+\varepsilon}{x}\Big)^{L-1}
\Big(\frac{Lx+\gamma}{L(x+\varepsilon)+\gamma}\Big)^{L-\alpha}\\
& = 1,
\end{align*}
we are in the presence of an absolutely outlier-prone model.

\section{Estimation by the Minimization of Stochastic Distances}
\label{estimation}

Information Theory provides divergences for comparing two distributions; in particular, we are interested in distances between densities.
Our proposal consists of, given the sample $\bm z$, computing $\widehat\alpha$, estimator of $\alpha$, as the point which minimizes the distance between the density $f_{\mathcal{G}_I^0}$ and an estimate of the underlying density function.

Cassetti et al.~\cite{APSAR2013ParameterEstimationStochasticDistances} assessed the Hellinger, Bhattacharyya, R\'enyi and Triangular distances, and they concluded that the latter outperforms the other ones in a variety of situations.
The Triangular distance between the densities $f_V$ and $f_W$ with common support $S$ is given by 
\begin{equation}
d_T(f_V,f_W)=\int_{S}\frac{(f_V-f_W)^2}{f_V+f_W}. \label{eq:TriangularDistance}
\end{equation}

Let $\bm{z}=(z_1,\dots, z_n)$ be a random sample of $n$ independent $\mathcal{G}_I^{0}(\alpha_0, \gamma^*_0, L_0)$-distributed observations. 
An  estimate of the underlying density function of $\bm{z}$, denoted $\widehat f$, is used to define the objective function to be minimized as a function of $\alpha$. 
The estimator for the $\alpha$ parameter based on the minimization of the Triangular distance between an estimate of the underlying density $\widehat f$ and the model $f_{\mathcal G_I^0}$, denoted by $\widehat{\alpha}_T$, is given by
\begin{equation}
\widehat{\alpha}_T= \arg\min_{-20\leq\alpha \leq -1} d_T\big(f_{\mathcal{G}_I^{0}}(\alpha,\gamma^*_0, L_0 ), \widehat f(\bm z)\big),
\label{minimization}
\end{equation}
where $\gamma_0^*$ and $L_0$ are known and $d_T$ is given in~\eqref{eq:TriangularDistance}.
Solving~\eqref{minimization} requires two steps: the integration of~\eqref{eq:TriangularDistance} using the $\mathcal G_I^0$ density and the density estimate $\widehat f$, and optimizing with respect to $\alpha$.
To the best of the authors' knowledge, there are no explicit analytic results for either problem, so we rely on numerical procedures.
The range of the search is established to avoid numerical instabilities.

We also calculate numerically the ML, the $\frac{1}{2}$-moment and Log-Cumulants based estimators using the same established range of the search  and compare all the methods through the bias and the mean squared error by simulation.

\subsection{Estimate of the Underlying Distribution}
\label{EmpiricalDistribution}

The choice of the way in which the underlying density function  is computed is very important in our proposal. 
In this section we describe several possibilities for computing it, and we justify our choice.

Histograms are the simplest empirical densities, but they lack unicity and they are not smooth functions.

Given that the model we are interested is asymmetric, another possibility is to use asymmetric kernels~\cite{chensx2000, ECT291329}.
Let $\bm z = (z_1,\dots, z_n)$ be a random sample of size $n$, with an unknown density probability function $f$, an  estimate of its density function using kernels is given by 
$$
\widehat{f}_b(t;\bm z)=\frac{1}{n}\sum_{i=1}^n K(t;z_i,b),
$$ 
where $b$ is the bandwith of the kernel $K$.

Among the many available asymmetric kernels, we worked with two: the densities of the Gamma and Inverse Gaussian distributions.
These kernels are given by
\begin{align*}
K_{\Gamma}( t; z_i,b) & =\frac{t^{{z_i}/{b}} \exp\{-{t}/{b}\}}{b^{{z_i}/{b}+1} \Gamma({z_i}/{b}+1)}, \text{and}\\
K_{\text{IG}}( t; z_i,b) & =\frac{1}{\sqrt{2\pi b t^3}} 
	\exp
		\Big\{-\frac{1}{2b z_i} \Big(\frac{t}{z_i}+\frac{z_i}{t}-2\Big)\Big\},
\end{align*}
respectively, for every $t>0$.
Empirical studies led us to employ $b=n^{-1/2}/5$.
 
As an example, Figure~\ref{figure:AjustesVarios2} shows the $\mathcal{G}_I^0(-3,2,1)$ density and three estimates of  the underlying density function obtained with $n=30$ samples: those produced by the Gamma and the Inverse Gaussian kernels, and the histogram computed with  the Freedman-Diaconis method. 

\begin{figure}[hbt]
 \begin{center}
 \includegraphics[width=\linewidth]{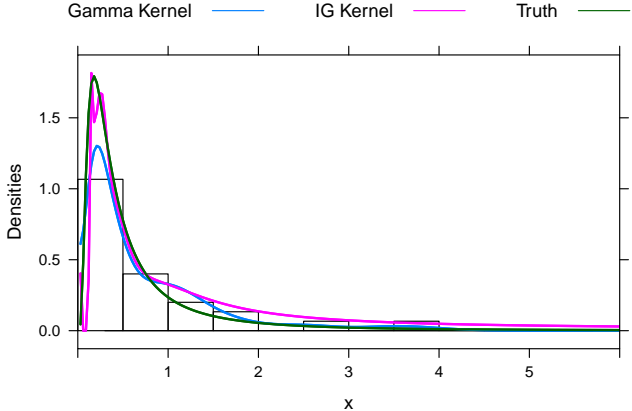} 
\caption{Fitting the $\mathcal{G}_I^0$ density of thirty ${\mathcal G_I^0}(-3,2,1)$ observations in different ways, along with the histogram.}
		 \label{figure:AjustesVarios2}
  \end{center}
\end{figure}

As it can be seen, fitting the underlying density function using kernels is better than using a histogram.
After extensive numerical studies, we opted for the $K_{\text{IG}}$ kernel due to its low computational cost, its ability to describe observations with large variance, and its good numerical stability.
In agreement with what Bouezmarni and Scaillet~\cite{ECT291329} reported, the Gamma kernel is also susceptible to numerical instabilities.

\subsection{Contamination}
\label{contamination}

Estimators in signal and image processing are often used in a wide variety of situations, so their robustness is of highest importance. 
Robustness is the ability to perform well when the data obey the assumed model, and to not provide completely useless results when the observations do not follow it exactly.
Robustness, in this sense, is essential when designing image processing and analysis applications.
Filters, for instance, employ estimators based, typically, in small samples which, more often than not, receive samples from more than one class.
Supervised classification relies on samples, and the effect of contamination (often referred to as ``training errors'') on the results has been attested, among other works, in~\cite{FreryFerreroBustosIJRS_FinalFinal2008}.	

In order to assess the robustness of the estimators, we propose three contamination models able to describe realistic departures from the hypothetical ``independent identically distributed sample'' assumption.

One of the sources of contamination in SAR imagery is the phenomenon of double bounce which results in some pixels having a high return value. 
The presence of such outliers may provoke big errors in the estimation. 

In order to assess the robustness of the proposal, we generate contaminated random samples using three types (cases) of contamination, with $0<\epsilon \ll 1$ the proportion of contamination. 
Let the Bernoulli random variable $B$ with probability of success $\epsilon$ model the occurrence of contamination.
Let $C \in \mathbb R_+$ be a large value.
\begin{itemize}
\item Case~1:
Let $W$   and $U$ be such that $W \sim \mathcal{G}_I^0(\alpha_1,\gamma_1^*,L)$,  and $U \sim \mathcal{G}_I^0(\alpha_2,\gamma_2^*,L) $. Define $Z=BU+(1-B)W$, then we generate $\{z_1,\dots,z_n\}$ identically distributed random variables with cumulative distribution function 
$$
(1-\epsilon) \mathcal{F}_{\mathcal{G}_I^0(\alpha_1,\gamma_1^*,L)}(z)+\epsilon\mathcal{F}_{\mathcal{G}_I^0(\alpha_2,\gamma_2^*,L)}(z),
$$
where $\mathcal{F}_{\mathcal{G}_I^0(\alpha,\gamma,L)}$ is the cumulative distribution function of a $\mathcal{G}_I^0(\alpha,\gamma,L)$ random variable.
\item Case~2: Consider $W \sim \mathcal{G}_I^0(\alpha_1,\gamma_1^*,L)$; return $Z=BC+(1-B)W$.
\item Case~3:
Consider $W \sim \mathcal{G}_I^0(\alpha,\gamma^*,L)$ and $U\sim \mathcal{G}_I^0(\alpha,10^k\gamma^*,L) $ with $k \in \mathbb{N}$. 
Return $Z=BU+(1-B)W$, then $\{z_1,\dots,z_n\}$ are identically distributed random variables with cumulative distribution function 
$$
(1-\epsilon) \mathcal{F}_{\mathcal{G}_I^0(\alpha,\gamma^*,L)}(z)+\epsilon\mathcal{F}_{\mathcal{G}_I^0(\alpha,10^k\gamma^*,L)}(z).
$$
\end{itemize}

All these models consider departures from the hypothesized distribution $\mathcal{G}_I^0(\alpha,\gamma^*,L)$.
The first type of contamination assumes that, with probability $\epsilon$, instead of observing outcomes from the ``right'' model, a sample from a different one will be observed; notice that the outlier may be close to the other observations.
The second type returns a fixed and typically large value, $C$, with probability $\epsilon$.
The third type is a particular case of the first, where the contamination assumes the form of a distribution whose scale is $k$ orders of magnitude larger than the hypothesized model.
We use the three cases of contamination models in our assessment.

\section{Results}\label{results}

A Monte Carlo experiment was set to assess the performance of each estimation procedure.
The parameter space consists of the grid formed by (i)~three values of roughness: $\alpha=\{-1.5, -3, -5\}$, which are representative of areas with extreme and moderate texture; (ii)~three usual levels of signal-to-noise processing, through $L=\{1,3,8\}$; (iii)~sample sizes $n=\{9, 25,49, 81,121,1000\}$, related to squared windows of side $3$, $5$, $7$, $9$ and  $11$, and to  a large sample; 
(iv)~each of the three cases of contamination, with $\epsilon=\{0.001,0.005,0.01\}$, $\alpha_2=\{-4,-15\}$, $C=100$ and $k=2$.

One thousand samples were drawn for each point of the parameter space, producing $\{\widehat{\alpha}_1, \dots, \widehat{\alpha}_{1000}\}$ estimates of each kind.
Estimates of the mean $\overline{\widehat{\alpha}}=(1000)^{-1}{\sum_{i=1}^{1000}{\widehat{\alpha}_i}}$, bias  $
\widehat{B}(\widehat\alpha) = \overline{\widehat\alpha_i}- \alpha$
and mean squared error $\widehat{\operatorname{mse}}=({1000})^{-1}{\sum_{i=1}^{1000}{(\widehat{\alpha}_i-\alpha)^2}}$ were then computed and compared.

In the following figures, ``ML'', ``T'',``Mom12'' and ``LCum'' denote the estimator based on the Maximum likelihood, Triangular distance, $\frac{1}{2}$-moment and Log-Cumulant, respectively.
Sample sizes are in the abscissas, which are presented in logarithmic scale.
The estimates of the mean are presented with error bars which show the Gaussian confidence interval at the $95\%$ level of confidence.
Only a few points of this parameter space are presented for brevity, but the remaining results are consistent with what is discussed here.

Figure~\ref{figure:alfasEstimados2} shows the mean of the estimators $\widehat{\alpha}$ in uncontaminated data ($\epsilon=0$) with different values of $n$ and $L$.
Only two of the four estimators lie, in mean, very close to the true value: $\widehat{\alpha}_{\text{MV}}$ and $\widehat{\alpha}_{\text{T}}$; it is noticeable how far from the true value lie $\widehat{\alpha}_{\text{LCum}}$ and $\widehat{\alpha}_{\text{Mom12}}$ when $\alpha=-5$.
It is noticeable that $\widehat\alpha_{\text{ML}}$ has a systematical tendency to underestimate the true value of $\alpha$.
Vasconcellos et al.~\cite{VasconcellosFrerySilva:CompStat} computed a first order approximation of such bias for a closely related model, and our results are in agreement with those.
The estimator based on the Triangular distance $\widehat\alpha_{\text T}$ compensates this bias.

Figure~\ref{figure:ecmEstimados2} shows the sample mean squared error of the estimates under the same situation, i.e., uncontaminated data. 
In most cases, all estimators have very similar $\operatorname{mse}$, not being possible to say that one is systematically the best one.
This is encouraging, since it provides evidence that $\widehat\alpha_{\text T}$ exhibits the first property of a good robust estimator, i.e., not being unacceptable under the true model.


\begin{figure}[hbt]
\centering
	{\includegraphics[angle=-90,width=\linewidth]{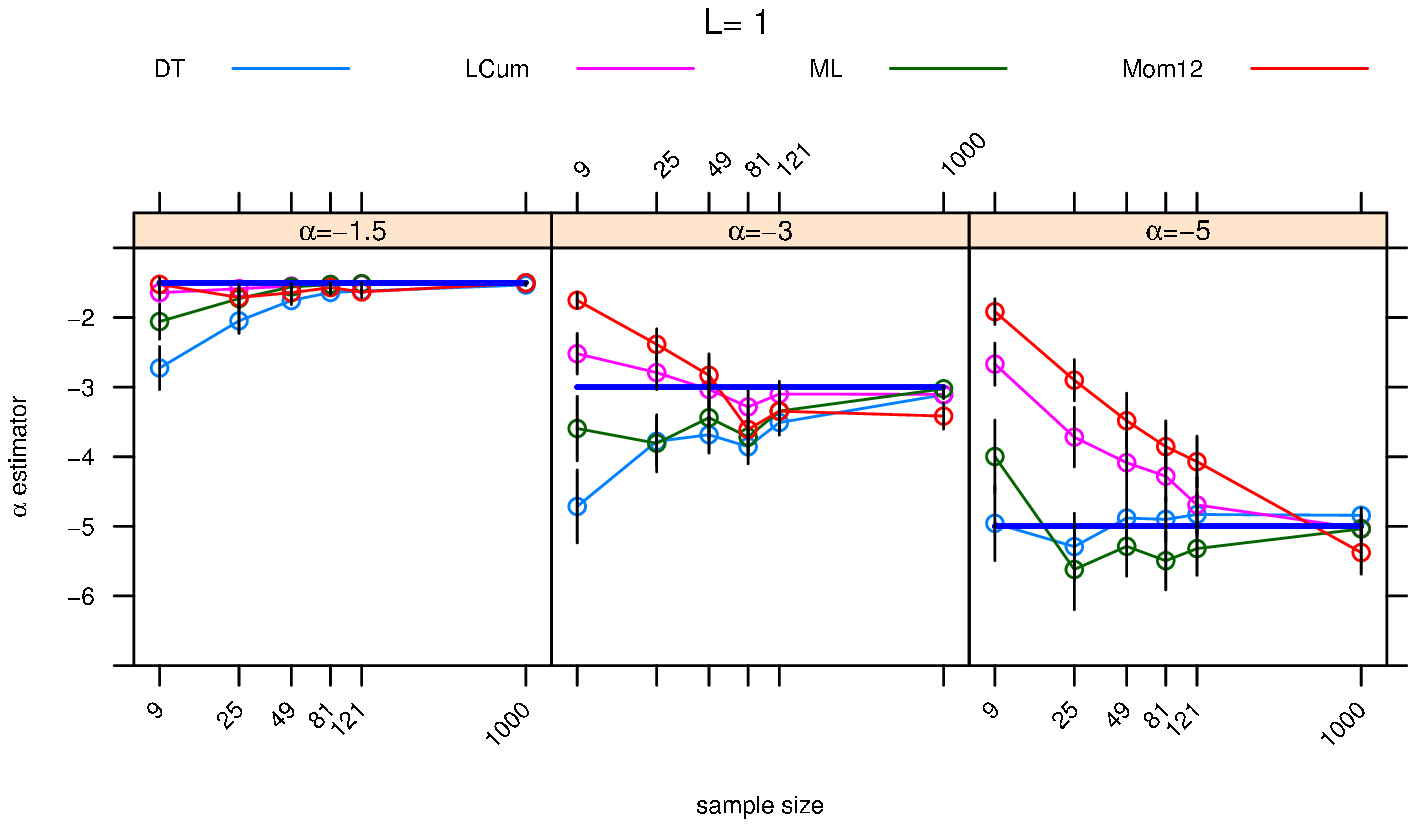} }\vskip 1em
	{\includegraphics[angle=-90,width=\linewidth]{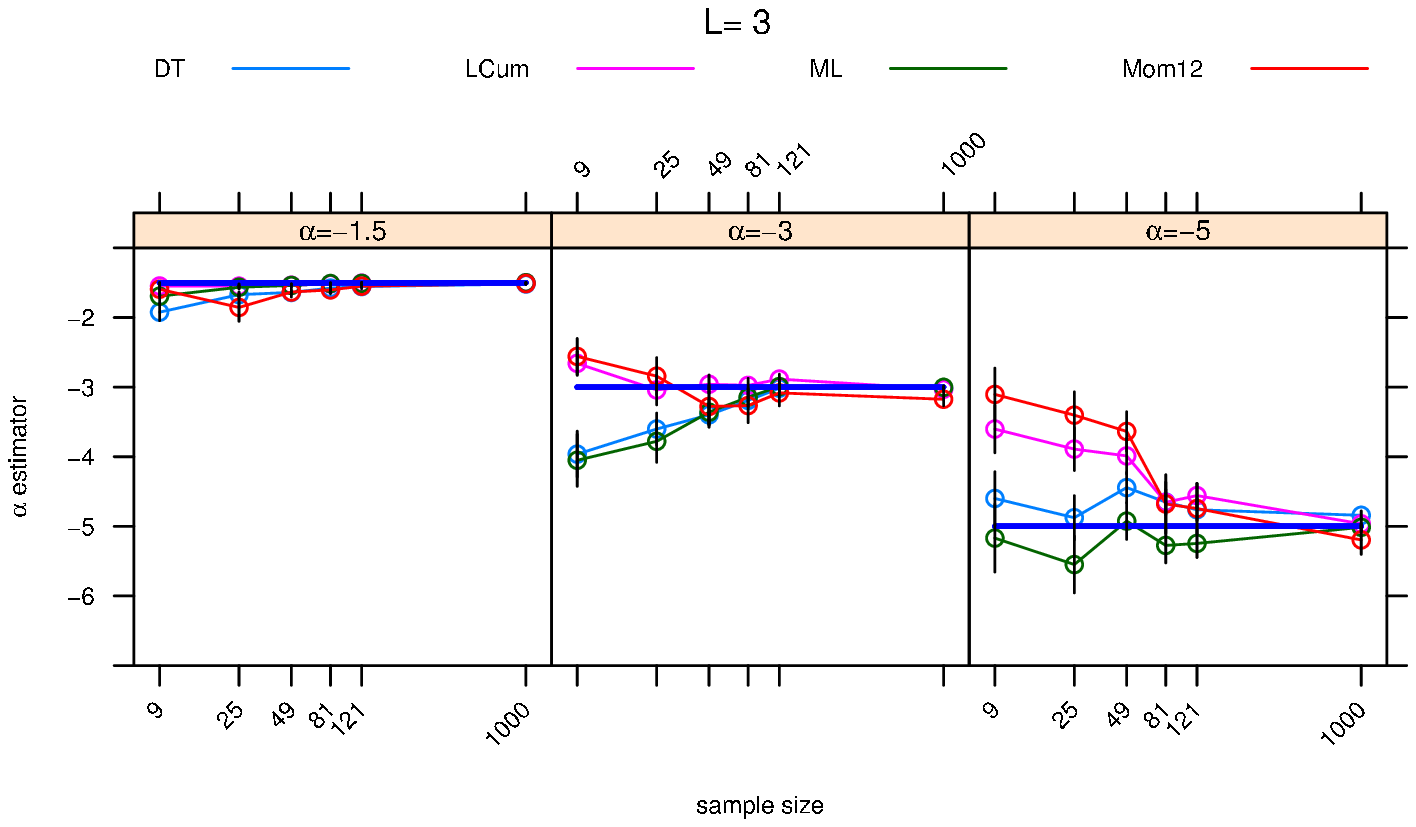}}\vskip 1em
	{\includegraphics[angle=-90,width=\linewidth]{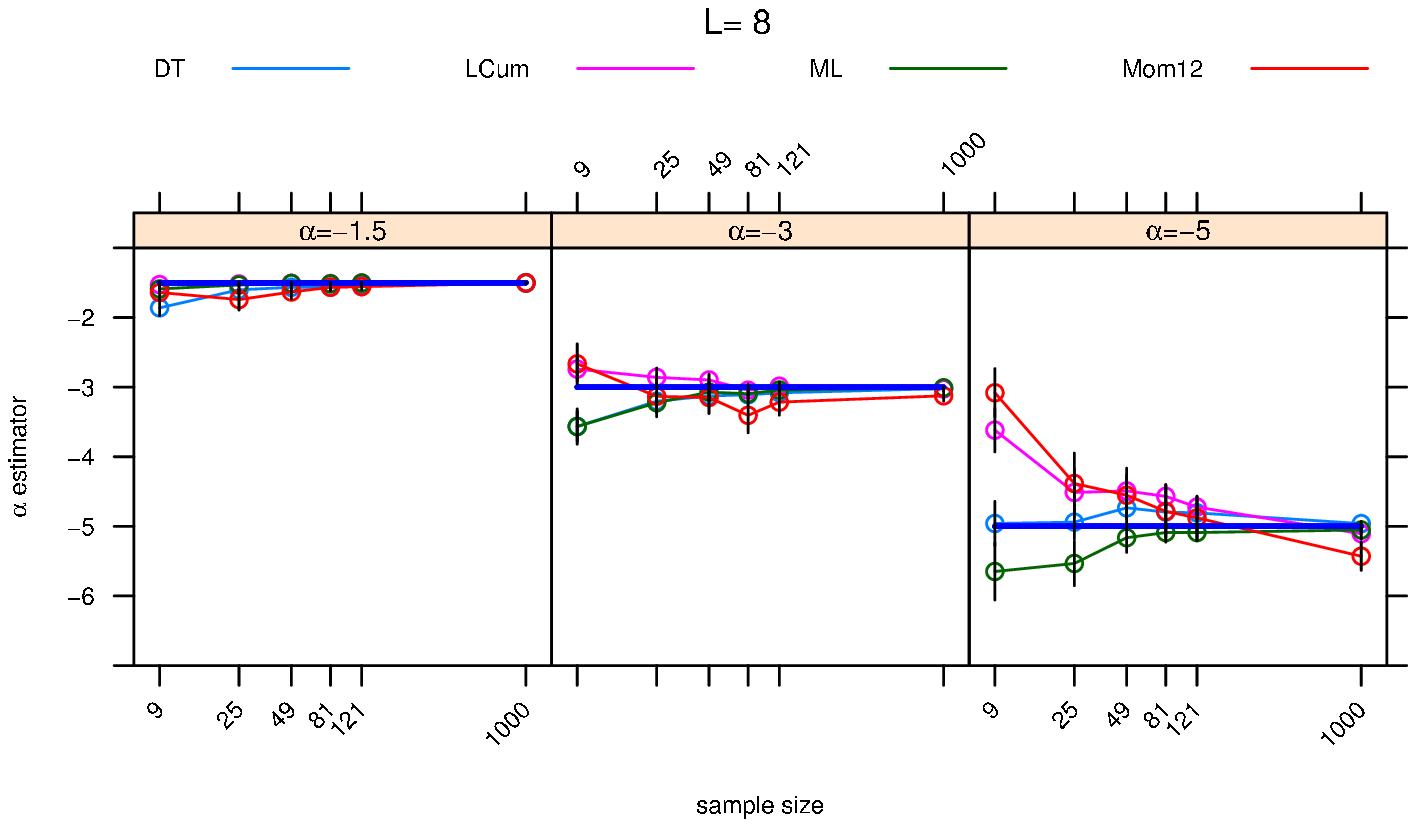}}
 \caption{Sample mean of estimates under uncontaminated data. The blue line is the true parameter.}
 \label{figure:alfasEstimados2}
\end{figure}

\begin{figure}[hbt]
\centering
 {\includegraphics[angle=-90,width=\linewidth]{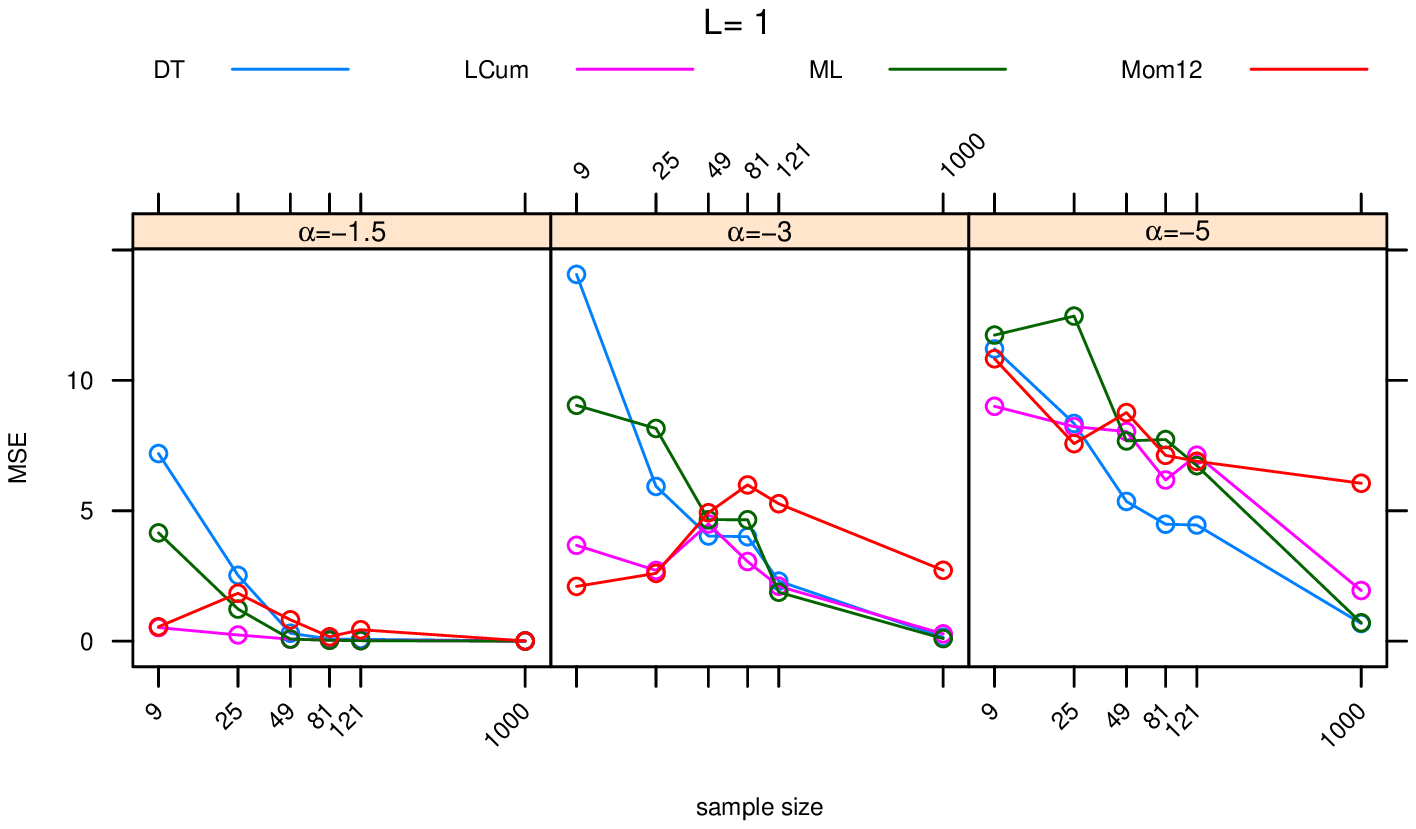} }\vskip 1em
	{\includegraphics[angle=-90,width=\linewidth]{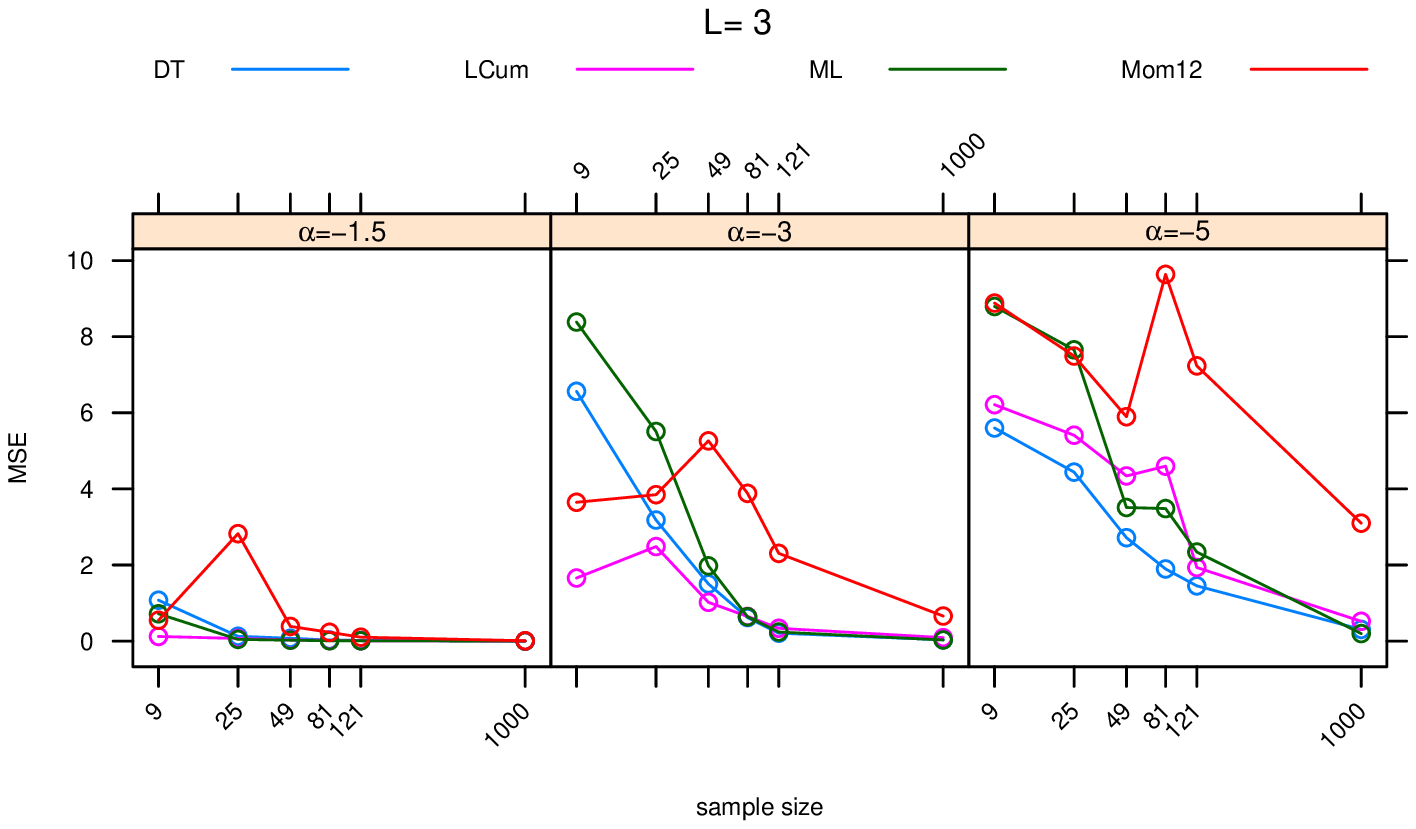}}\vskip 1em
	{\includegraphics[angle=-90,width=\linewidth]{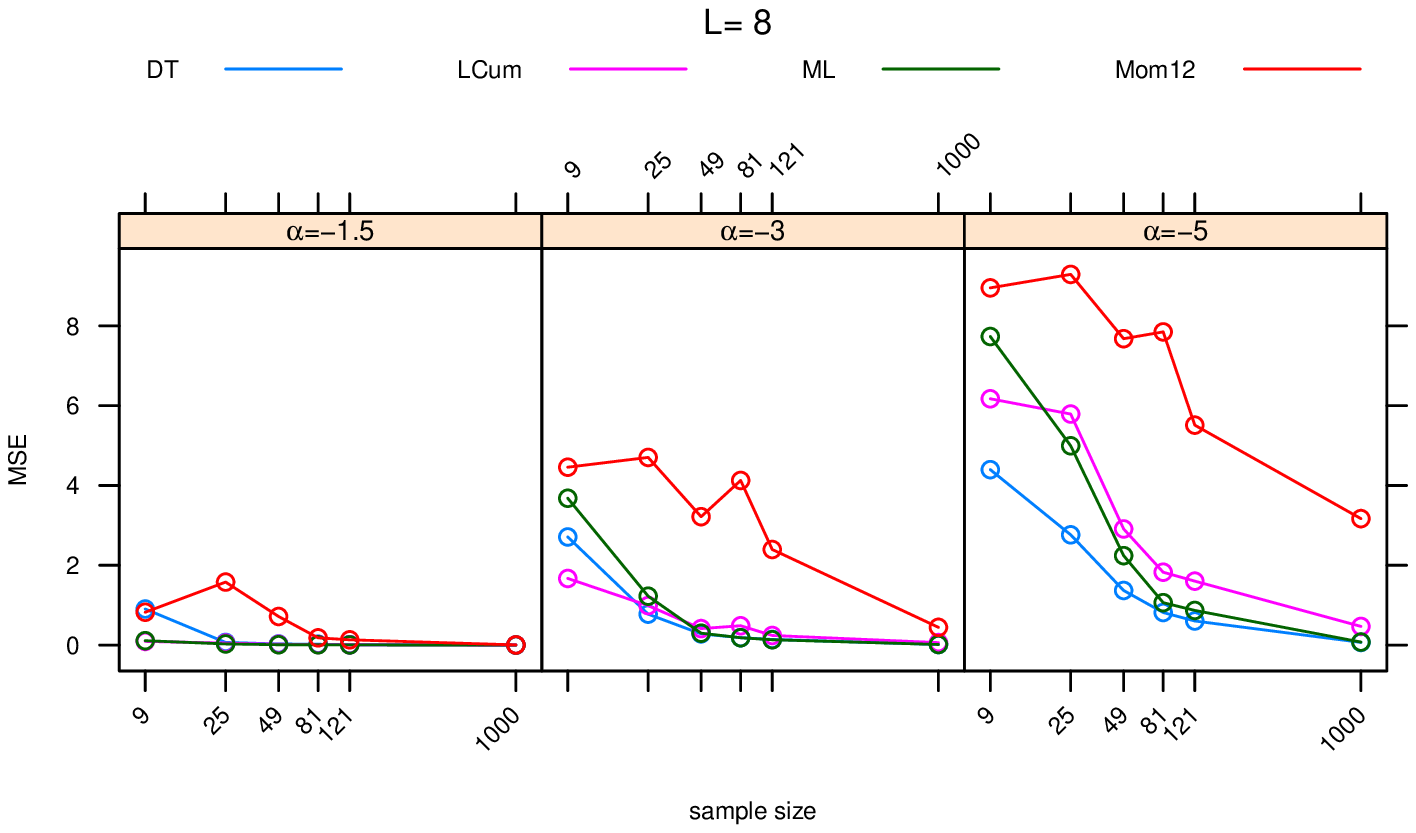}}
 \caption{Sample mean squared error of estimates under uncontaminated data.}
 \label{figure:ecmEstimados2}
\end{figure}

The mean times of processing, measured in seconds, for each method and each case, were computed, as an example, the mean times of processing for $L=1$ and $n=81$ are presented in Table~\ref{tablaDeTiemposmedios}. It can be seen that the new method has a higher computational cost. The other cases are consistent with this table.
The details of the computer platform are presented in the appendix.

\begin{table}[htb]
\caption{Mean Times for simulated data without Contamination, $L=1$, $n=81$ }
\label{tablaDeTiemposmedios}
\centering
\begin{tabular}{cccc}
\toprule
MV& DT& Moment $\frac{1}{2}$ & Log Cumulant \\
\midrule
$0.003$& $2.223$ & $0.0001$ &$0.003$ \\
\bottomrule
\end{tabular}
\end{table}


Figures~\ref{figure:alfasEstimados3} and~\ref{figure:ecmEstimados3} show, respectively, the sample mean and mean squared error of the estimates under Case~1 contamination with $\alpha_2=-15$, $\epsilon=0.01$, and varying $n$ and $L$.  
This type of contamination injects, with probability $\epsilon=0.01$, observations with almost no texture in the sample under analysis.
As expected, the influence of such perturbation is more noticeable in those situations where the underlying model is further away from the contamination, i.e., for larger values of $\alpha$.
This is particularly clear in Fig.~\ref{figure:ecmEstimados3}, which shows that the mean squared errors of $\widehat\alpha_{\text{ML}}$, $\widehat\alpha_{\text{Mom12}}$ and $\widehat\alpha_{\text{LCum}}$ are larger than that of $\widehat\alpha_{\text T}$ for $L=3,8$, with not a clear distinction for $L=1$ except that $\widehat\alpha_{\text T}$ is at least very competitive in the $\alpha=-3,-5$ cases.

\begin{figure}[hbt] 
\centering
 {\includegraphics[angle=-90,width=\linewidth]{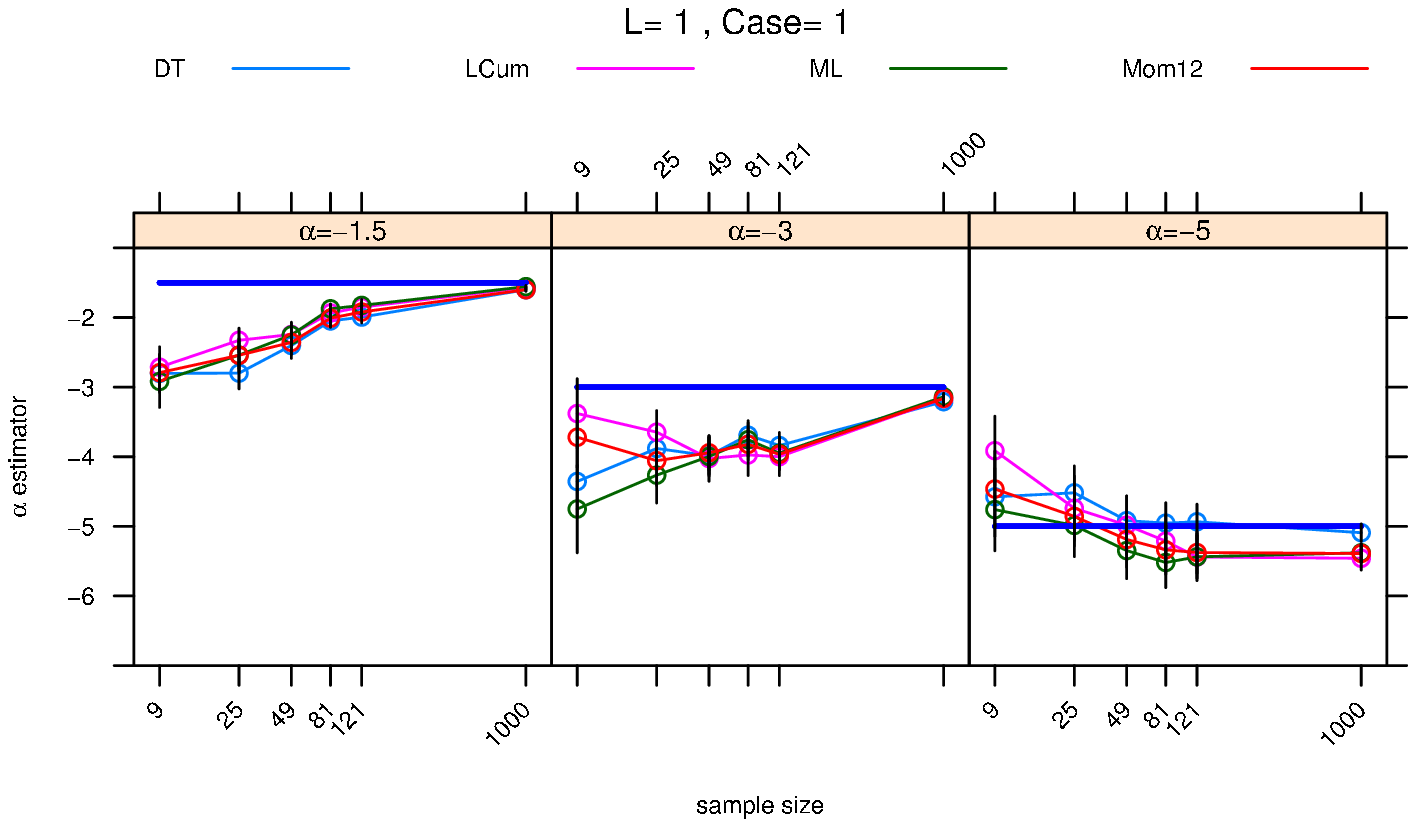} }\vskip 1em
	{\includegraphics[angle=-90,width=\linewidth]{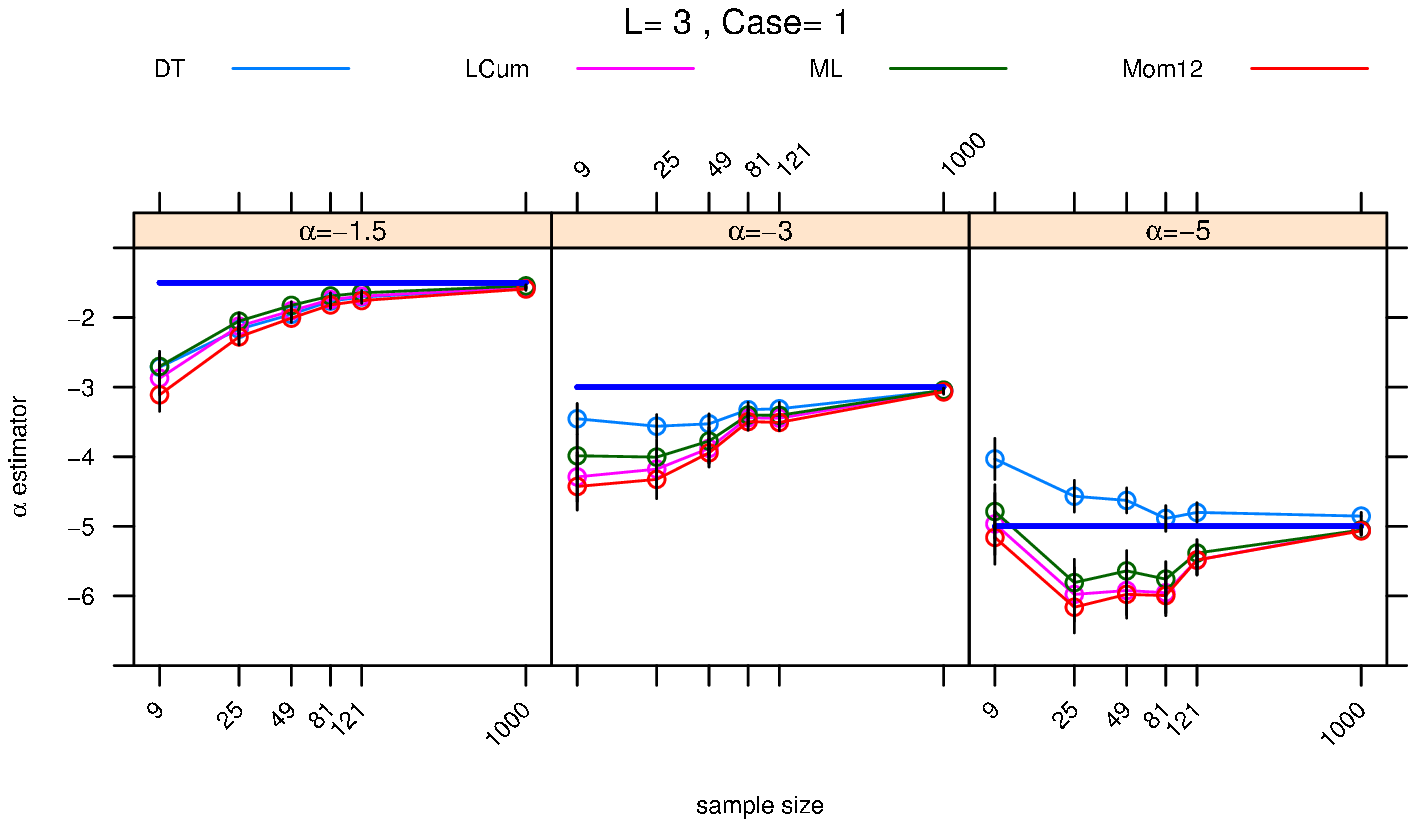}}\vskip 1em
	{\includegraphics[angle=-90,width=\linewidth]{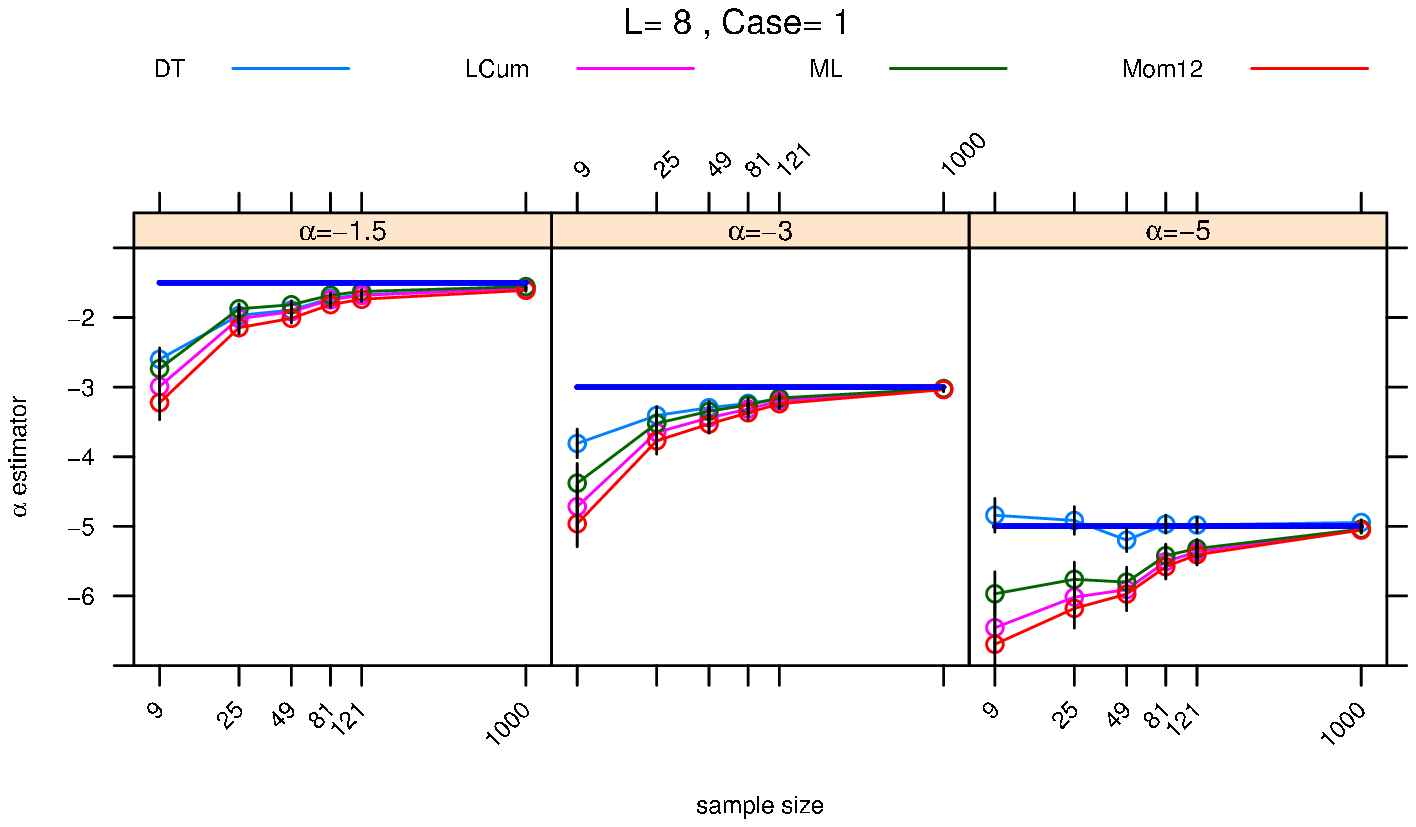}}
 \caption{Sample mean of estimates, Case~1 with $\alpha_2=-15$ and $\epsilon=0.01$.}
 \label{figure:alfasEstimados3}
\end{figure}

\begin{figure}[hbt]
\centering
 {\includegraphics[angle=-90,width=\linewidth]{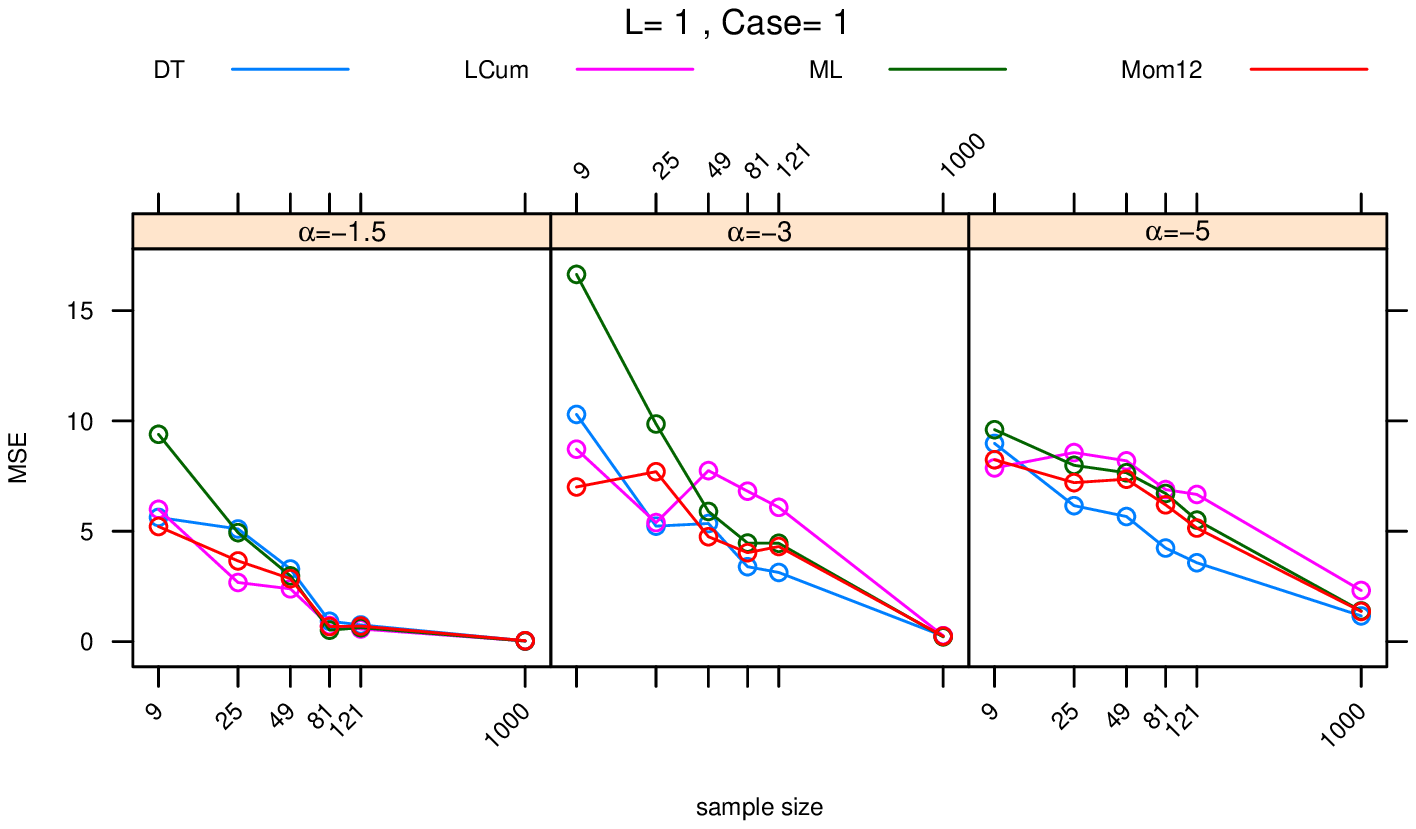} }\vskip 1em
	{\includegraphics[angle=-90,width=\linewidth]{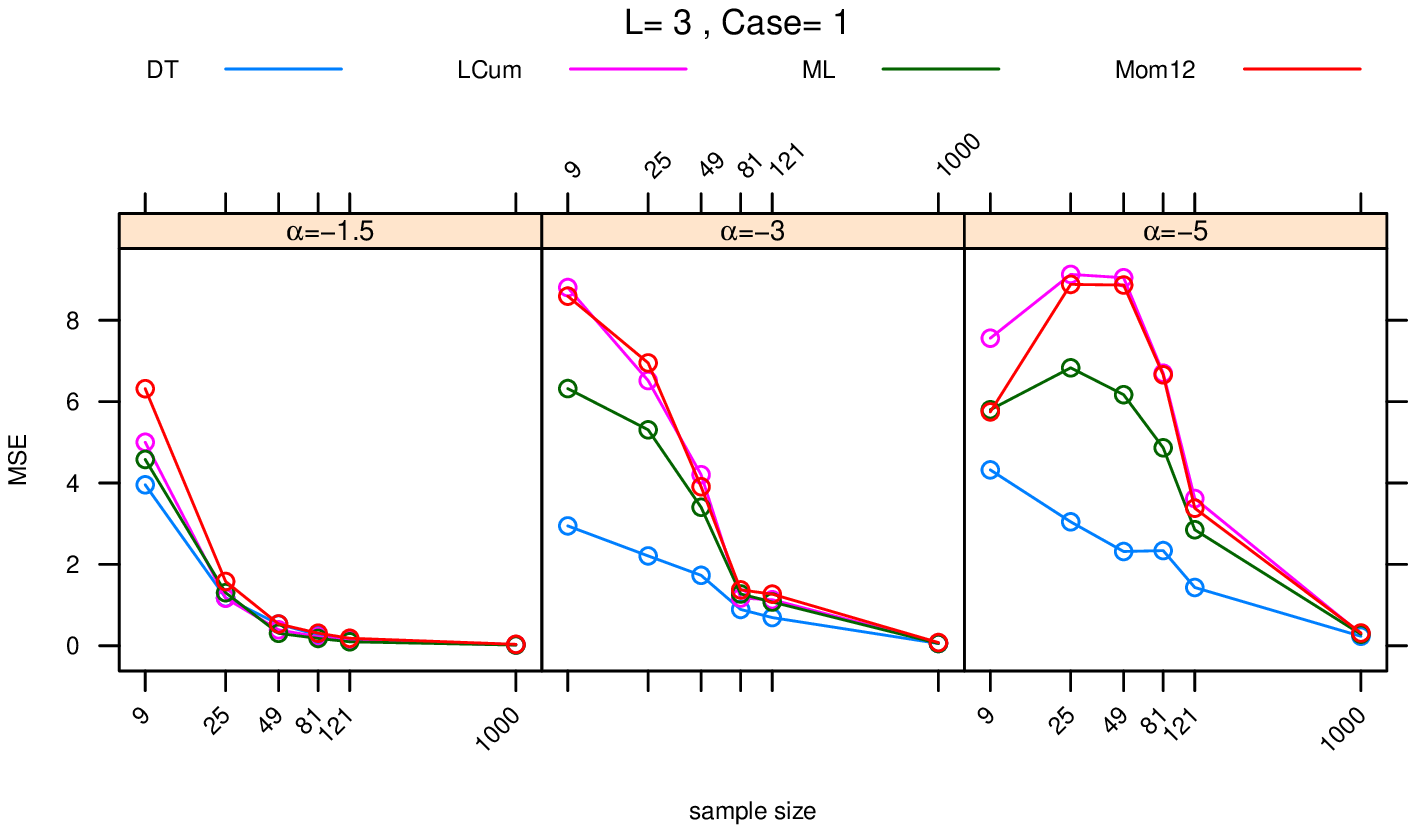}}\vskip 1em
	{\includegraphics[angle=-90,width=\linewidth]{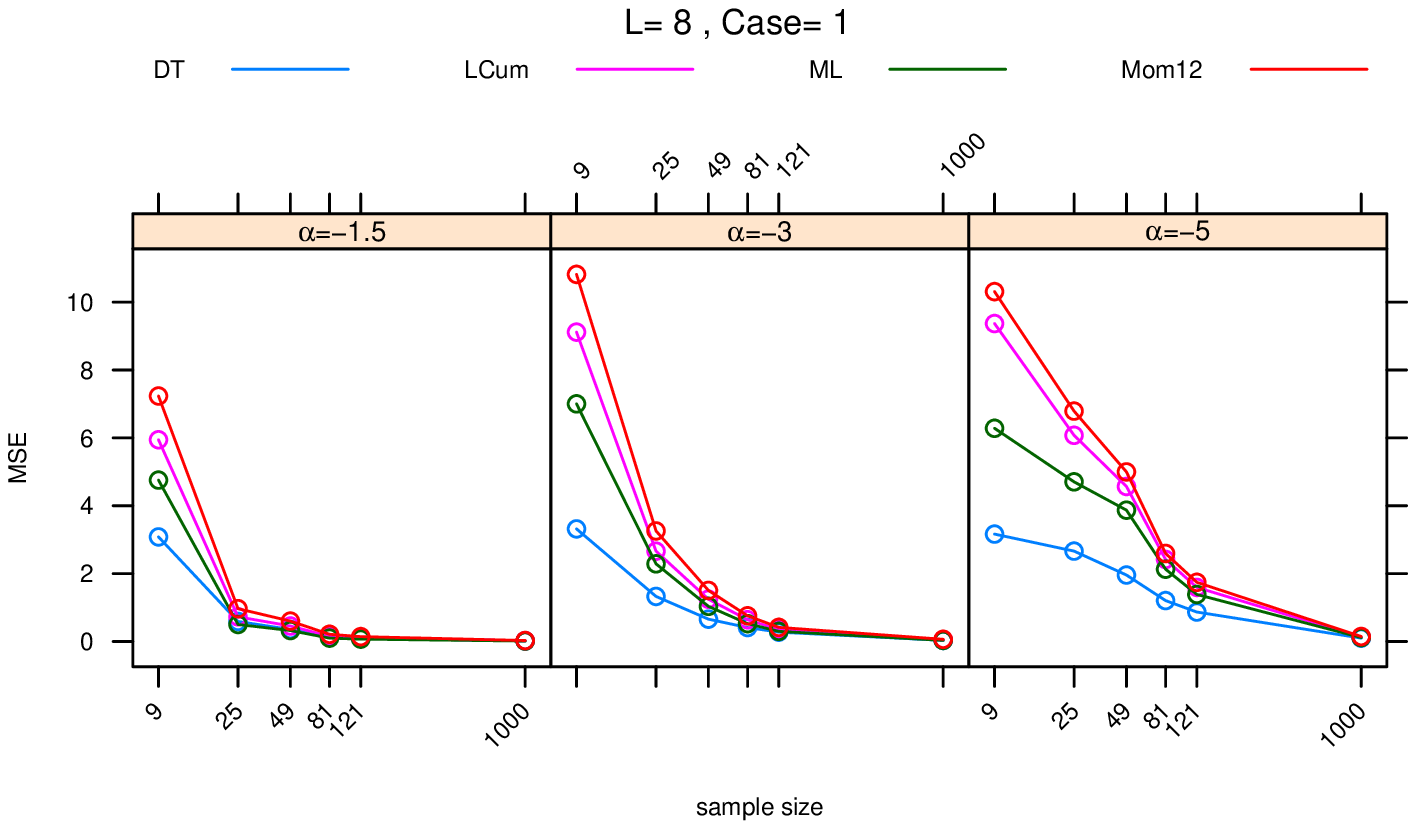}}
 \caption{Sample mean squared error of estimates, Case~1 with $\alpha_2=-15$ and $\epsilon=0.01$. }
 \label{figure:ecmEstimados3}
\end{figure}


Figures~\ref{figure:alfasEstimados4} and~\ref{figure:ecmEstimados4} present, respectively, the sample mean and sample mean squared error of the estimates under Case~2 contamination with $\epsilon=0.001$ and $C=100$.
This type of contamination injects a constant value ($C=100$) with probability $\epsilon=0.001$ instead of each observation from the $\mathcal G_I^0$ distribution.
Since we are considering samples with unitary mean, this is a large contamination.
In this case, $\widehat\alpha_{\text T}$ is, in mean, closer to the true value than the other methods, and its mean square error the smallest.
 
\begin{figure}[htb]
\centering
 {\includegraphics[angle=-90,width=\linewidth]{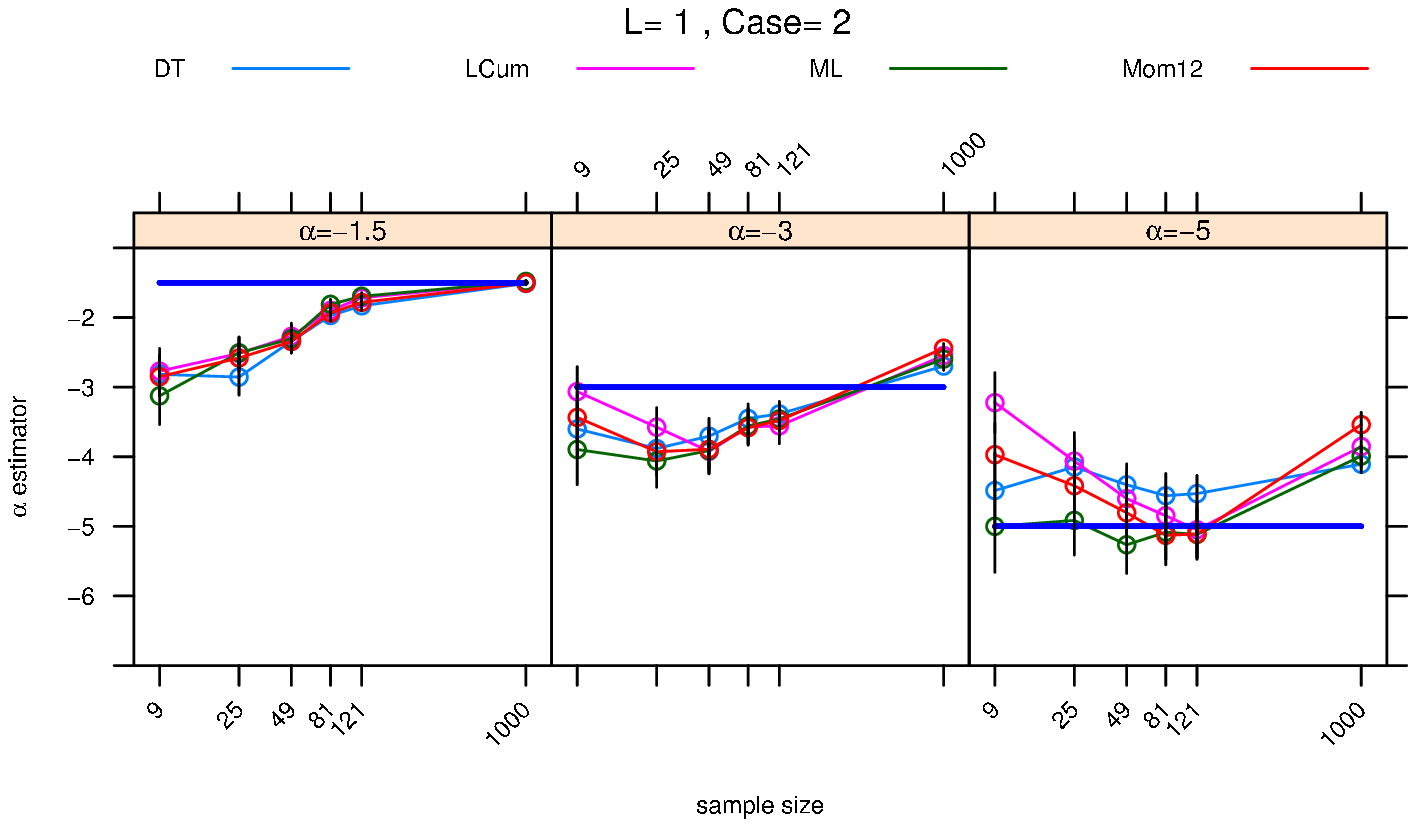} }\vskip 1em
	{\includegraphics[angle=-90,width=\linewidth]{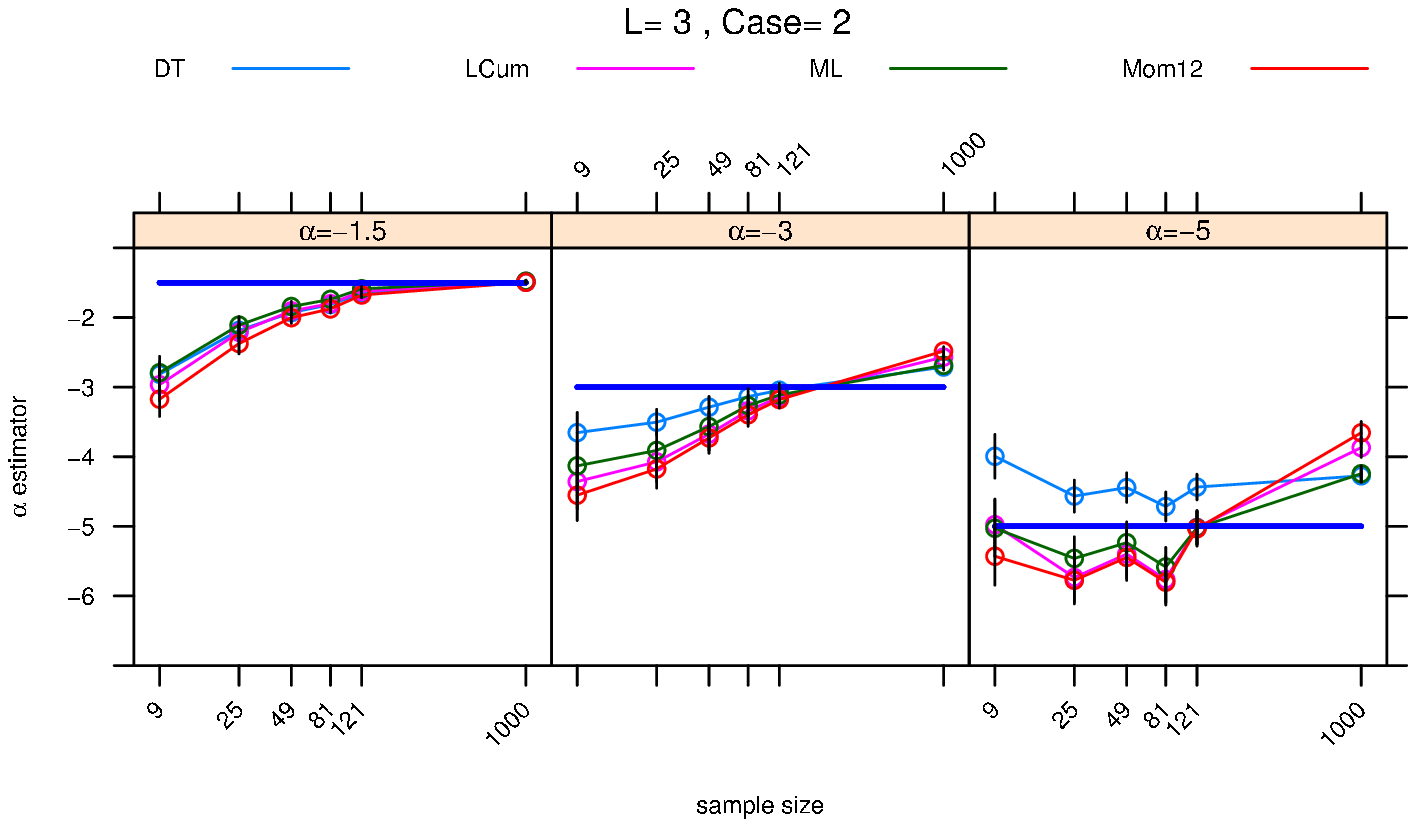}}\vskip 1em
	{\includegraphics[angle=-90,width=\linewidth]{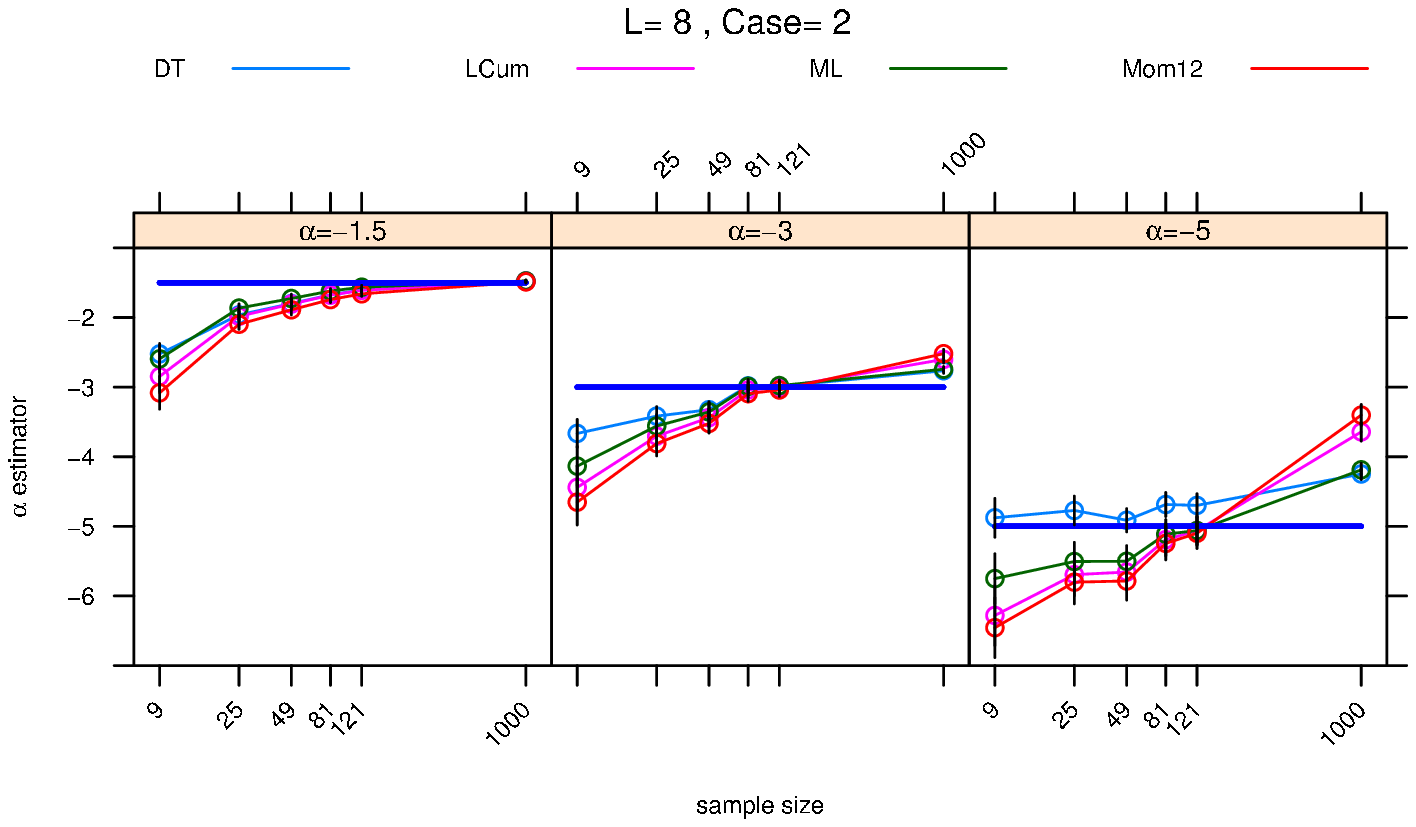}}
 \caption{Sample mean of estimates, Case~2 with $C=100$, $\epsilon=0.001$.}
 \label{figure:alfasEstimados4}
\end{figure}

\begin{figure}[htb]
\centering
 {\includegraphics[angle=-90,width=\linewidth]{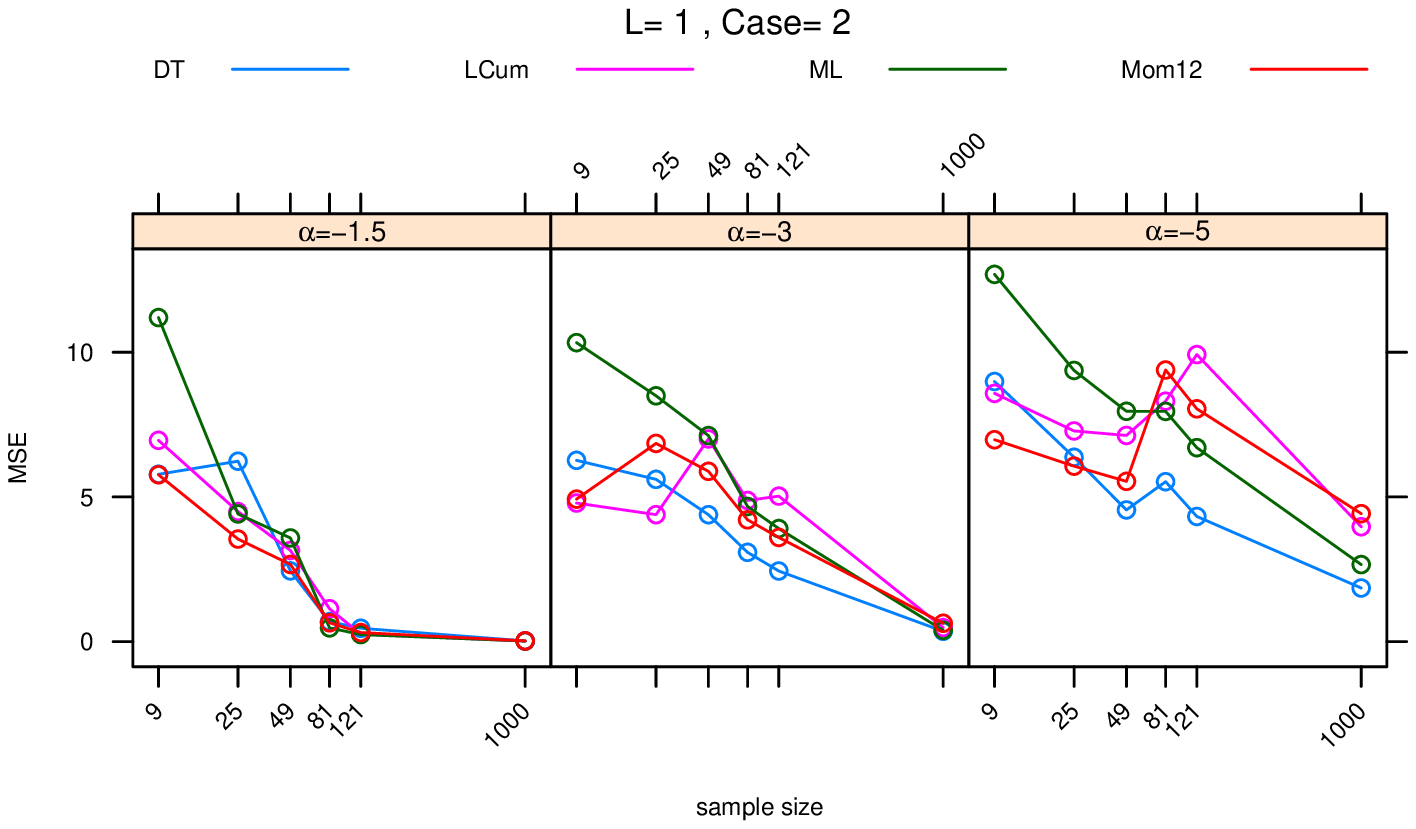} }\vskip 1em
	{\includegraphics[angle=-90,width=\linewidth]{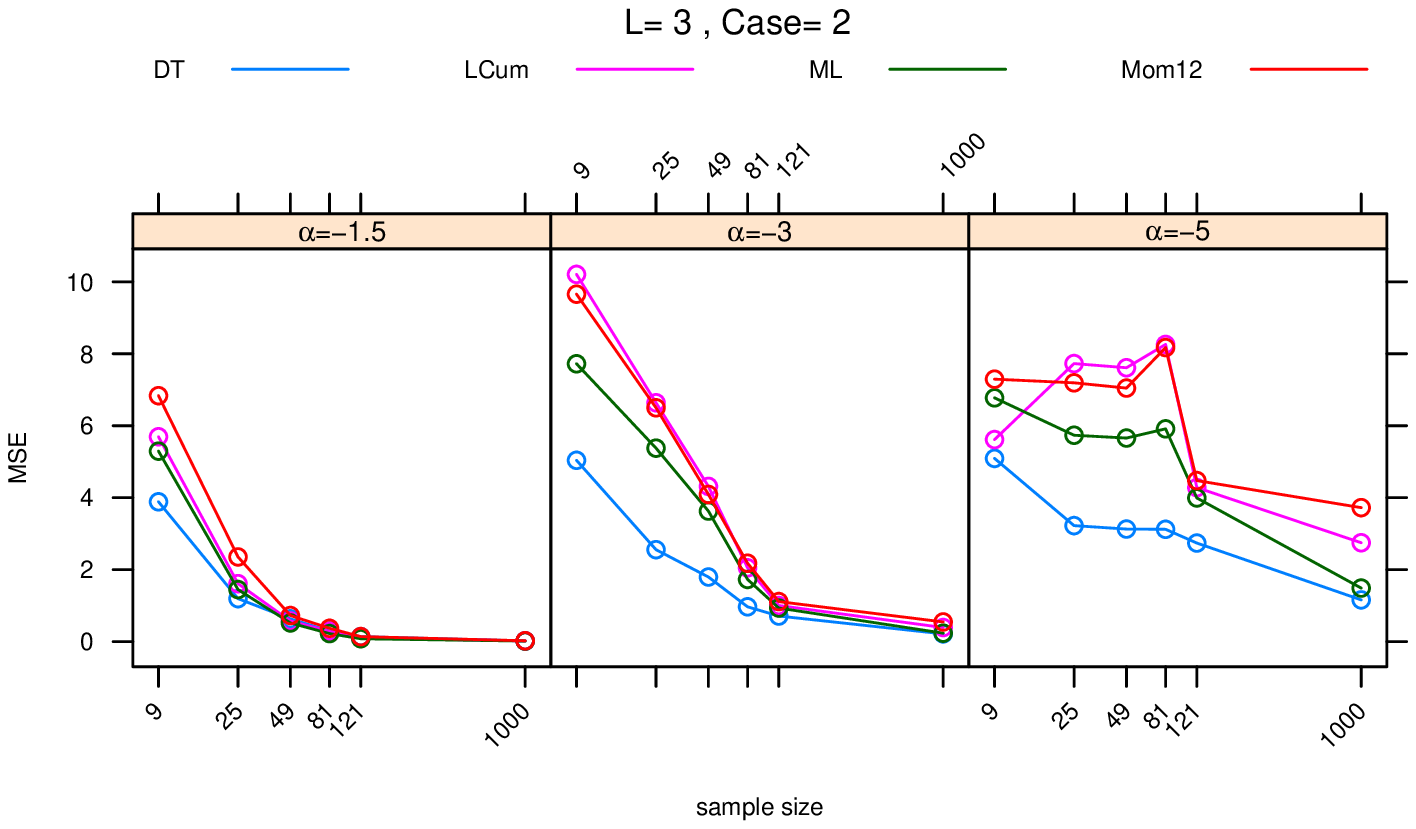}}\vskip 1em
	{\includegraphics[angle=-90,width=\linewidth]{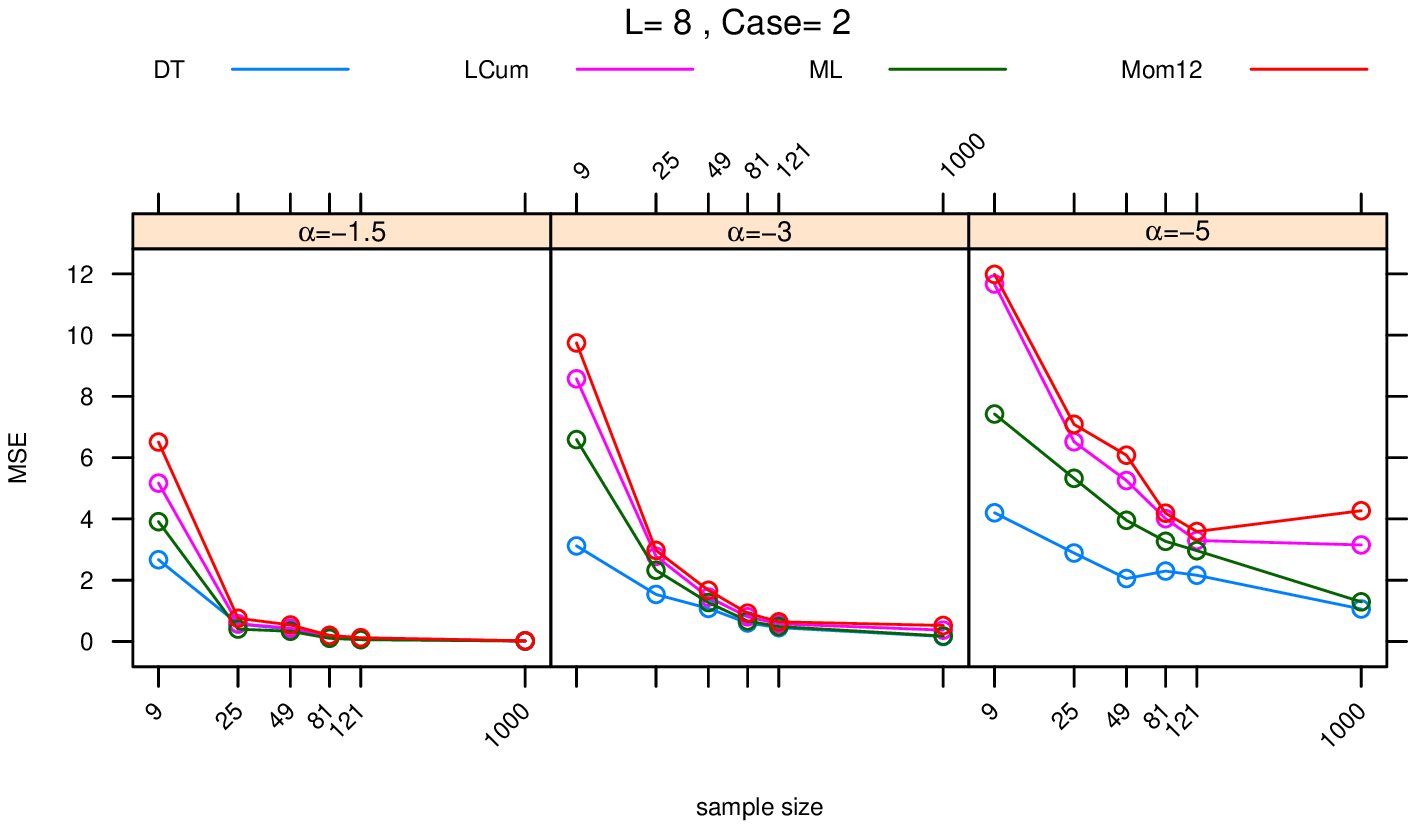}}
 \caption{Sample mean squared error of estimates, Case~2 with $C=100$, $\epsilon=0.001$.}
 \label{figure:ecmEstimados4}
\end{figure}


Figures~\ref{figure:alfasEstimados5} and~\ref{figure:ecmEstimados5} show, respectively, the sample mean and mean squared error of the estimates under Case~3 with $\epsilon=0.005$ with $k=2$.
This kind of contamination draws, with probability $\epsilon=0.005$, an observation from a $\mathcal G_I^0$ distribution with a scale one hundred times larger than that of the ``correct'' model.
The behavior of the estimators follows the same pattern for $L=3,8$: $\widehat\alpha_{\text T}$ produces the closest estimates to the true value with reduced mean squared error.
There is no good estimator for the single-look case with this case of contamination.

\begin{figure}[htb]
\centering
 {\includegraphics[angle=-90,width=\linewidth]{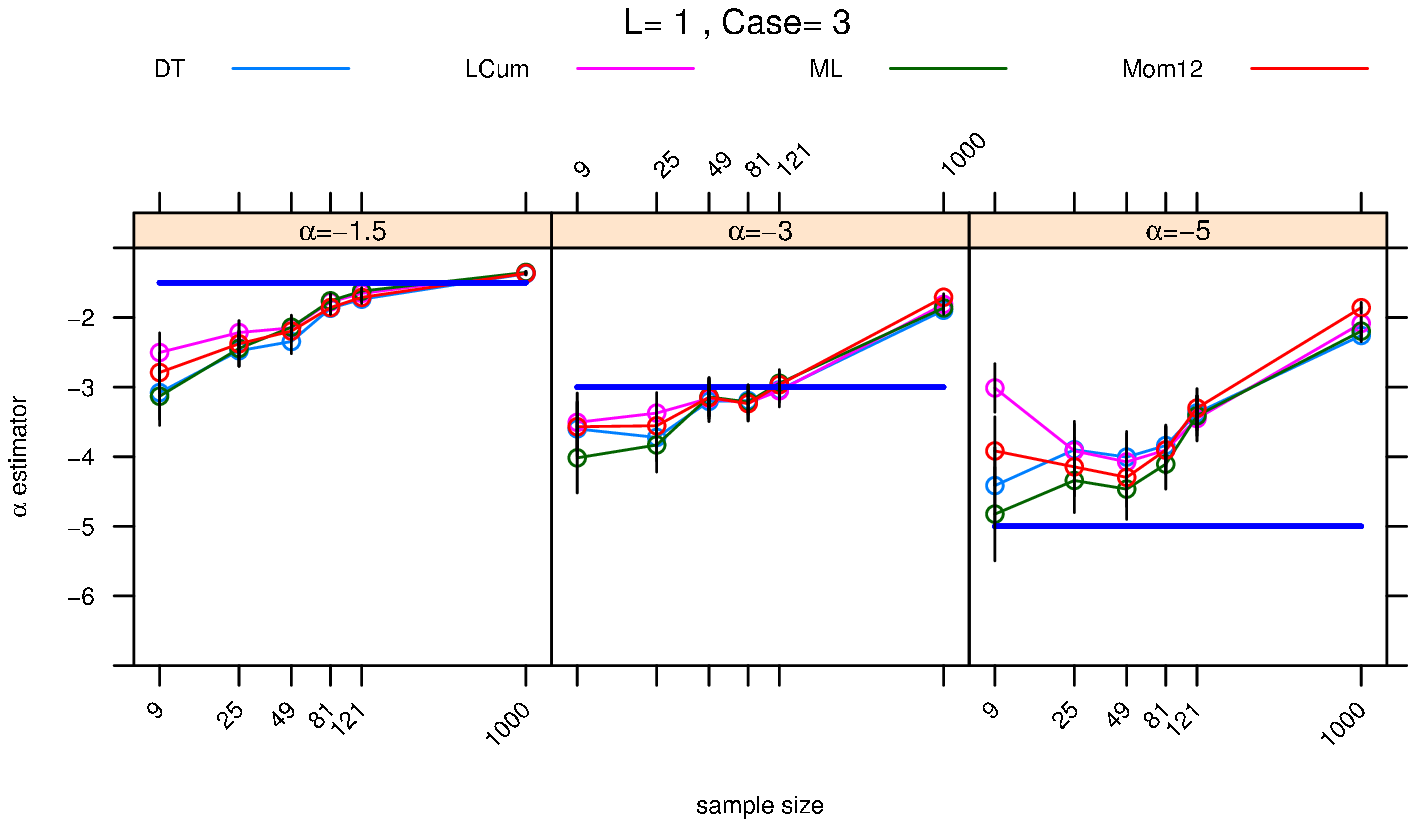} }
	{\includegraphics[angle=-90,width=\linewidth]{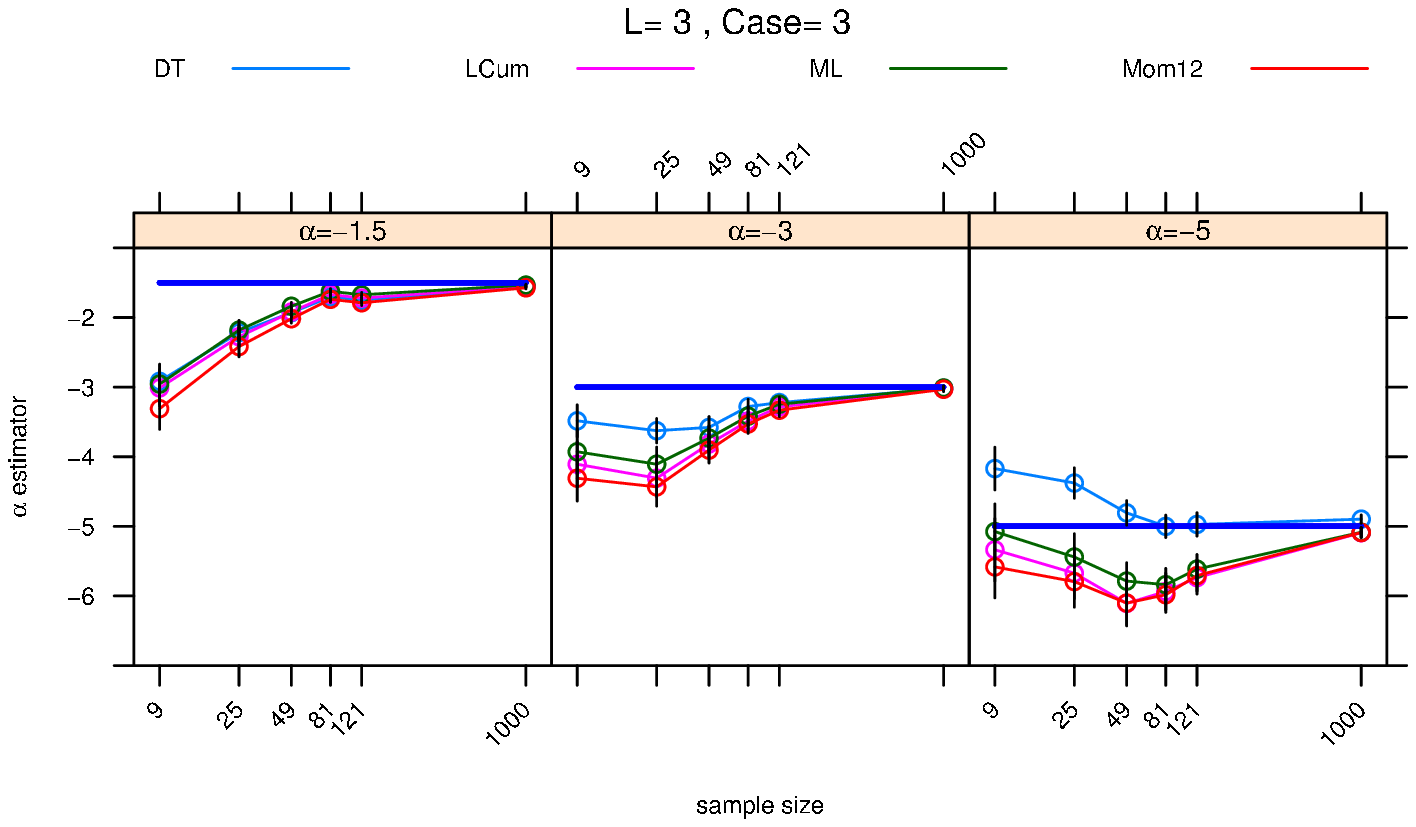}}
	{\includegraphics[angle=-90,width=\linewidth]{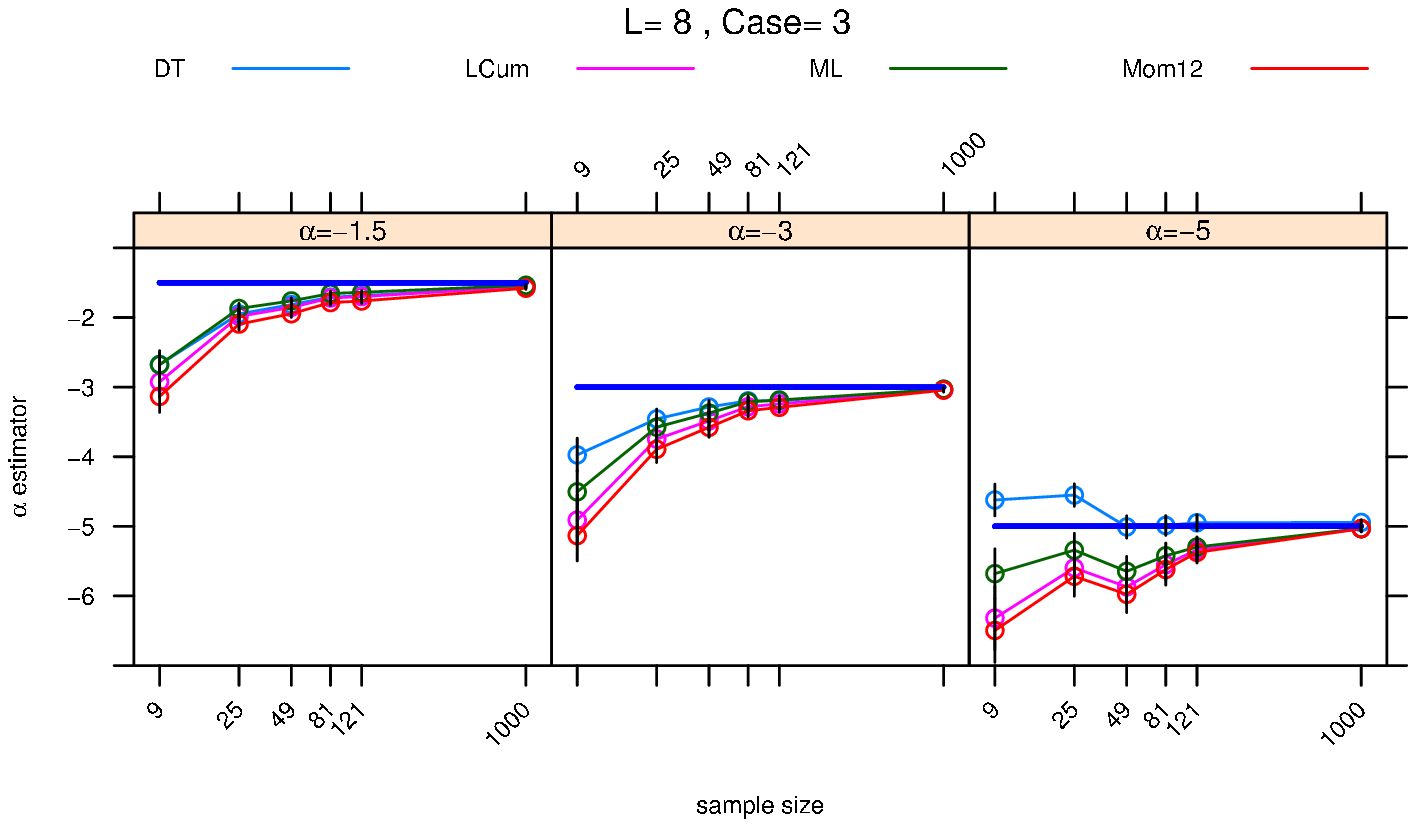}}
 \caption{Sample mean of estimates, Case~3 with $\epsilon=0.005$ and $k=2$.}
 \label{figure:alfasEstimados5}
\end{figure}

\begin{figure}[htb]
\centering
 {\includegraphics[angle=-90,width=\linewidth]{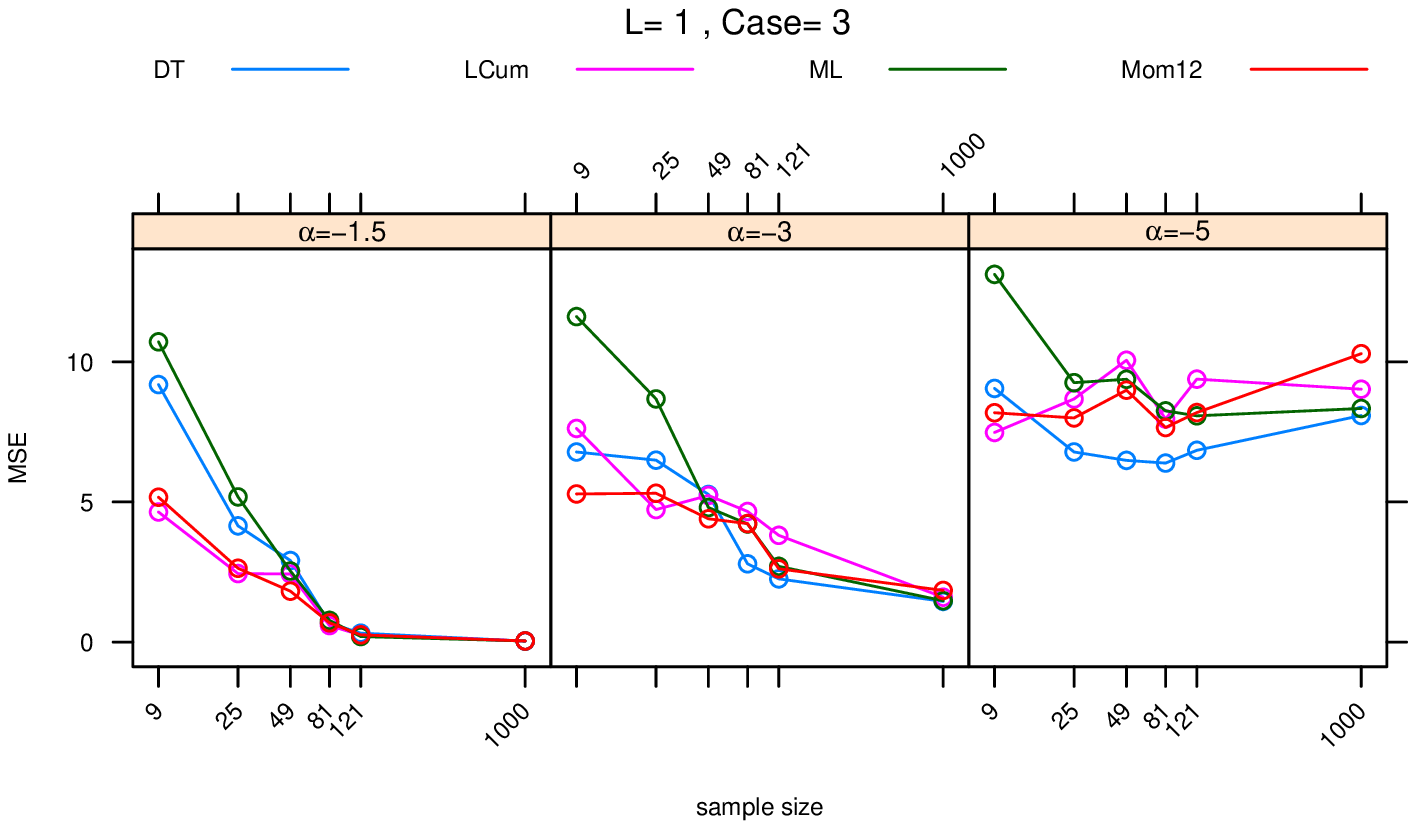} }\vskip 1em
	{\includegraphics[angle=-90,width=\linewidth]{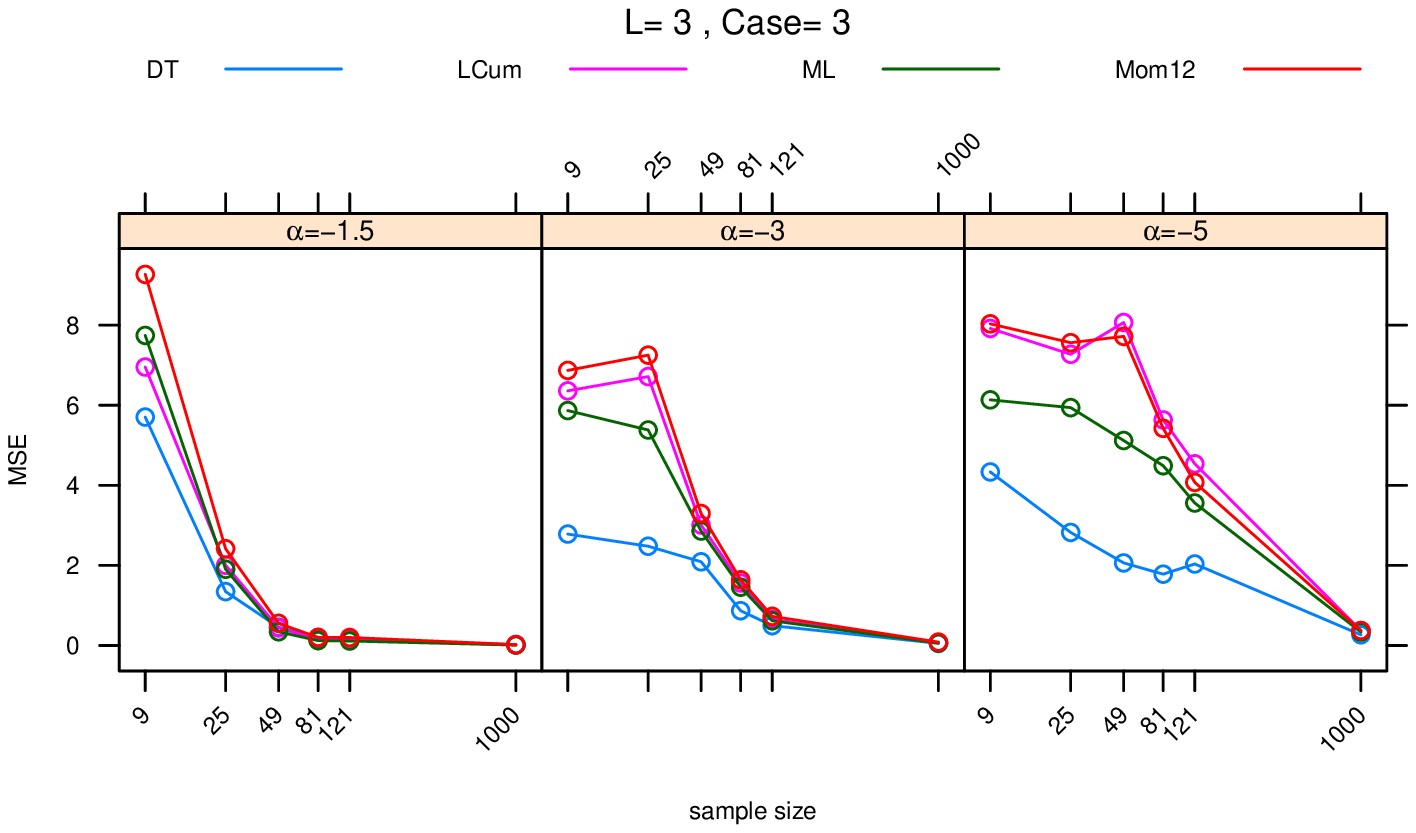}}\vskip 1em
	{\includegraphics[angle=-90,width=\linewidth]{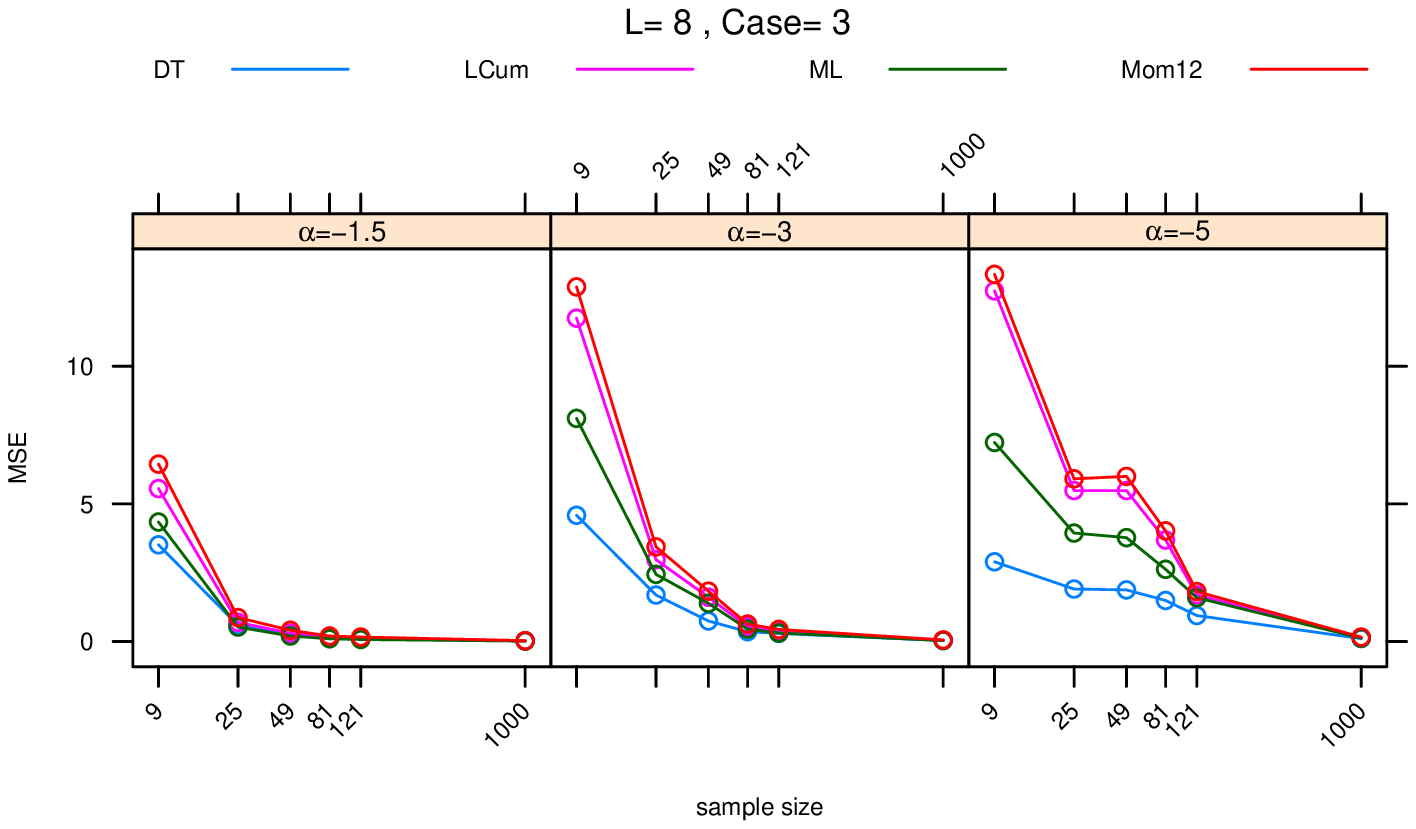}}
 \caption{Sample mean squared error of estimates, Case~3 with $\epsilon=0.005$ and $k=2$.}
 \label{figure:ecmEstimados5}
\end{figure}

As previously mentioned, the iterative algorithms which implement the $\frac{1}{2}$-moment and Log-Cumulant estimators often fail to converge. In the Monte Carlo experiment, if $\frac{1}{2}$-moment or Log-Cumulant fail, we eliminate the estimators computed with the other methods in the corresponding iteration, so  the amount of elements for calculating the mean $\overline{\widehat{\alpha}}$, the bias and the mean squared error is lower than $1000$.
As an example, Table~\ref{tabla_removidos} informs the number of such situations for Case~1 and $\alpha_2 = -15, \epsilon = 0.01$.
These results are consistent with other situations under contamination, and they suggest that these methods are progressively more prone to fail in more heterogeneous areas.

\begin{table}[htb]
\centering
\caption{Percentage of situations for which no convergence was observed for the Moments and Log-Cumulant methods in Case~1, $\alpha_2 = -15, \epsilon = 0.01$}\label{tabla_removidos}
\begin{tabular}{crr}
\toprule
 $L$ & $\frac{1}{2}$-Moment  & Log-Cumulant \\
 \midrule
$1$ & $22.87$  & $21.56$ \\
$3$ & $11.71$  &  $11.87$ \\
$8$ & $5.81$ & $6.04$  \\
\bottomrule
\end{tabular}
\end{table}

Data from a single-look L-band HH polarization in intensity format E-SAR~\cite{ESAR} image was used in the following.
Figure~\ref{figure:reales2} shows the regions used for estimating the texture parameter.
Table~\ref{tablaResultadosMunich} shows the results of estimating the $\alpha$ parameter for each rectangular region, where $NA$ means that the corresponding estimator is not available.

\begin{figure}[hbt]
\centering
\includegraphics[viewport=200 0 450 500, clip=true, angle=-90, width=\linewidth]{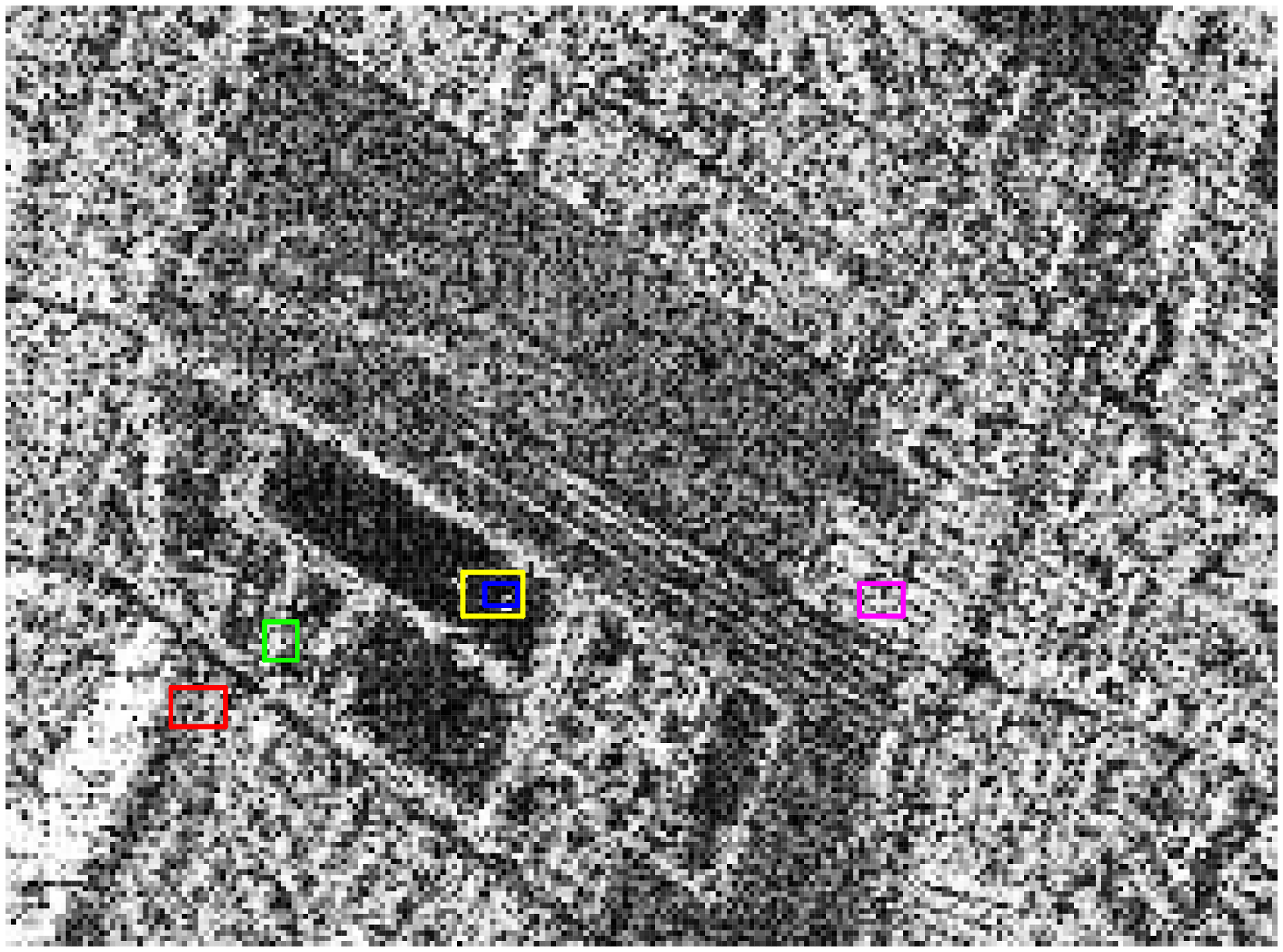}
\caption{Real SAR image and the regions used to estimate the $\alpha$-parameter.}\label{figure:reales2}
\end{figure}


The Kolmogorov-Smirnov test (KS-test) is applied to two samples: $\bm x$, from the image, and $\bm y$, simulated.
Babu and Feigelson~\cite{Jogesh2006} warn about the use of the same sample for parameter estimation and for performing a KS test between the data and the estimated cumulative distribution function.
We then took the real sample $\bm x$, used to estimate the parameters with the four methods under assessment, and then samples of the same size $\bm y$ were drawn from the $\mathcal G_I^0$ law with those parameters.
The KS test was then performed between $\bm x$ and $\bm y$ for the null hypothesis $H_0$ ``both samples come from the same distribution,'' and the complementary alternative hypothesis.

Table~\ref{resultadosTestMunich} shows the sample $p$-values. 
It is observed that the null hypothesis is not rejected with a significance level of $5$\% in any of the cases.
This result justifies the adequacy of the model for the data.


Fig.~\ref{figure:reales3} shows the regions used for estimating the texture parameter under the influence of a corner reflector. 
Table~\ref{resultadosCorner} shows the estimates of $\alpha$ for each rectangular region in the image of Fig.~\ref{Cornerregiones}.
The Maximum Likelihood, $\frac{1}{2}$-moment and Log-Cumulant estimators are unable to produce any estimate in small samples.
They only produce sensible values when the sample is of at least $90$ observations, 
The Log-Cumulant method is the one which requires the largest sample size to produce an acceptable estimate.
The estimator based on the Triangular distance yields plausible values under contamination even with very small samples.

\begin{figure}[htb]
\centering
\subfigure[Single-look E-SAR image with a corner reflector.]{\includegraphics[angle =90,width=.8\linewidth]{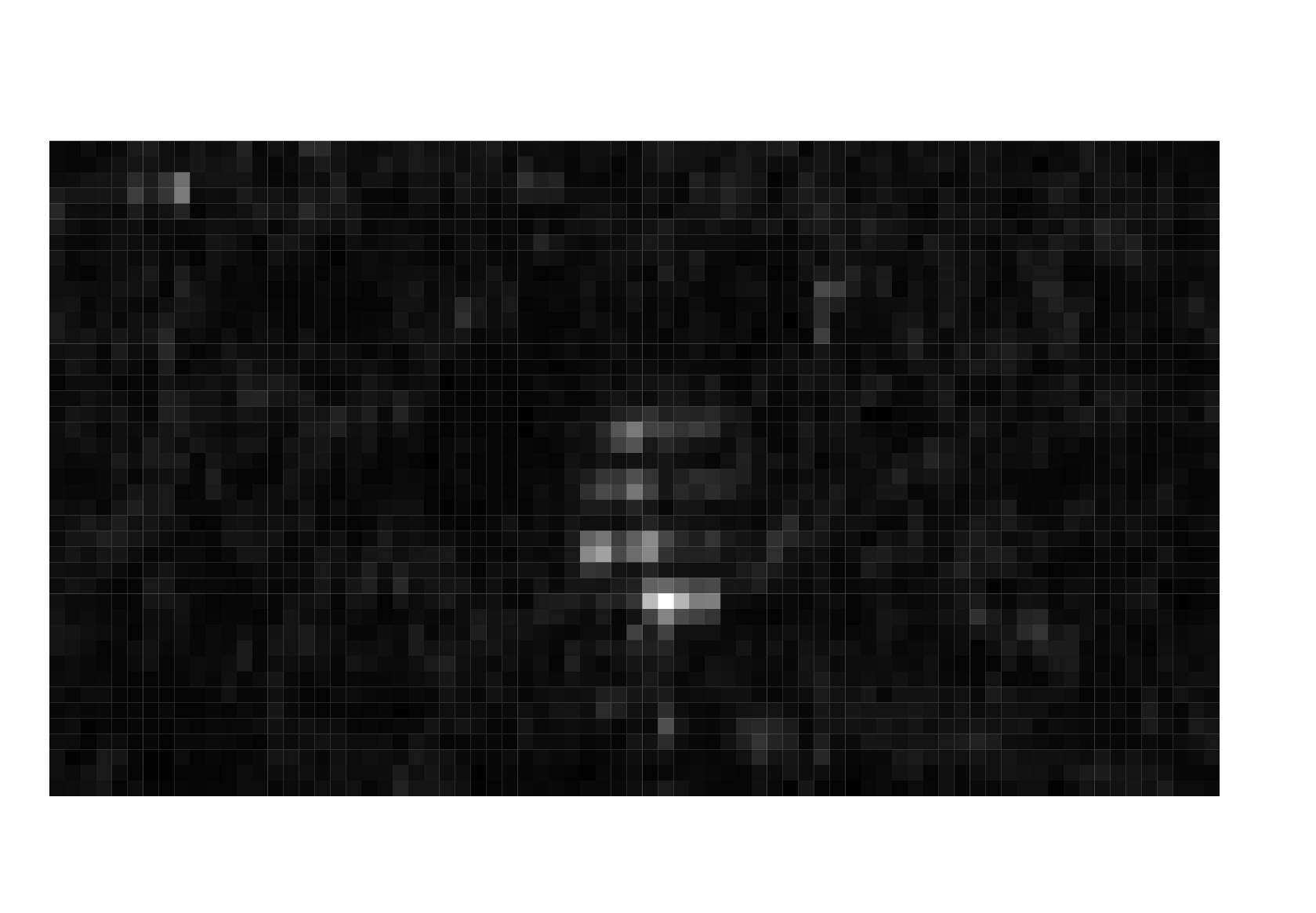}}
\subfigure[Regions of interest.]{\includegraphics[angle =90,width=.8\linewidth,]{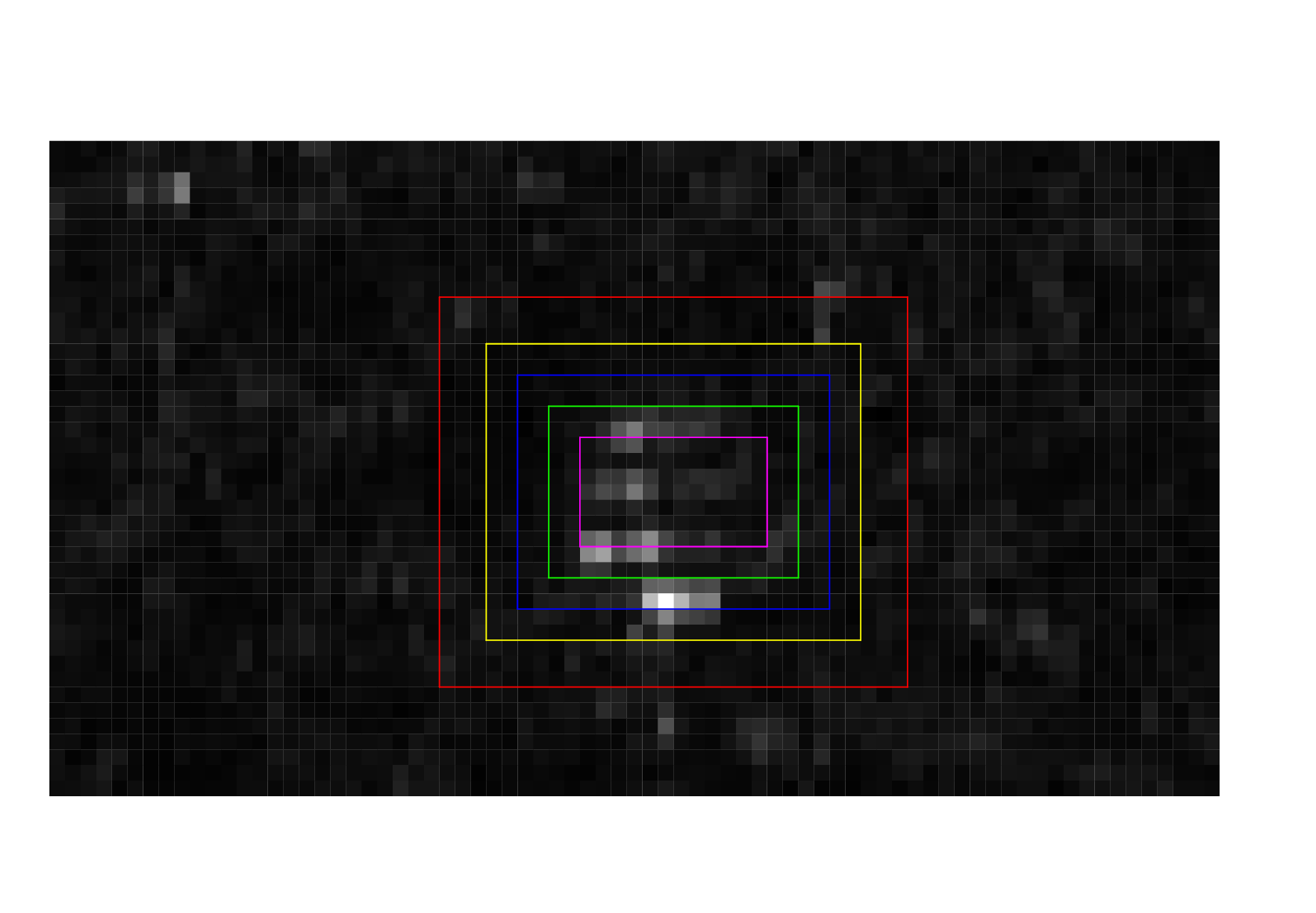}\label{Cornerregiones}}
\caption{Samples of several sizes in a real SAR image with a corner reflector, used to estimate the $\alpha$-parameter.}\label{figure:reales3}
\end{figure}


Table~\ref{pvaluesTrueNestedSamples} presents the $p$-values of the KS test, applied as previously described to the samples of Fig.~\ref{Cornerregiones}.
Estimators based on the $\frac{1}{2}$-moments and on Log-Cumulants failed to produce estimates in two samples and, therefore, it was not possible to apply the procedure.
The other samples did not fail to pass the KS test, except for the Blue sample and the Maximum Likelihood estimator.
These results leads us to conclude that the underlying model is a safe choice regardless the presence of contamination and the estimation procedure.


\section{Conclusions}
\label{conclusions}

We proposed a new estimator for the texture parameter of the $\mathcal{G}_I^0$ distribution based on the minimization of a stochastic distance between the model (which is indexed by the estimator) and an estimate of the probability density function built with asymmetric kernels.
We defined three models of contamination inspired in real situations in order to assess the impact of outliers in the performance of the estimators.
The difficulties of estimating the textured parameter of the $\mathcal{G}_I^0$ distribution were further investigated and justified throughout new theoretical results regarding its heavytailedness.

Regarding the impact of the contamination on the performance of estimators, we observed the following.
Only Fractional Moment and Log-Cumulant estimators fail to converge.
\begin{itemize}
\item Case~1: Regardless the intensity of the contamination, the bigger the number of looks, the smaller the percentage of situations for which no convergence if achieved.
\item Cases~2 and~3: the percentage of situations for which there is no convergence increases with the level of contamination, and reduces with $\alpha$.
\end{itemize}
In the single-case look the proposed estimator does not present excellent results, but it never fails to converge whereas the others are prone to produce useless output.

The new estimator presents good properties as measured by its bias and mean squared error. 
It is competitive with the Maximum Likelihood, Fractional Moment and Log-Cumulant estimators in situations without contamination, and outperforms the other techniques even in the presence of small levels of contamination.

For this reason, it would be advisable to use $\widehat{\alpha}_{\text{T}}$ in every situation, specially when small samples are used and/or when there is the possibility of having contaminated data.
The extra computational cost incurred in using this estimator is, at most, the twenty times the required to compute $\widehat{\alpha}_{\text{ML}}$, but its advantages outnumber these extra computer cycles.

\appendix

Simulations were performed using the \texttt R language and environment for statistical computing~\cite{R} version~3.0.2.
The \texttt{adaptIntegrate} function from the \texttt{cubature} package was used to perform the numerical integration required to evaluate the Triangular distance, the algorithm utilized is an adaptive multidimensional integration over hypercubes. 

In order to numerically find $\widehat\alpha_{\text{LCum}}$  we used the function \texttt{uniroot} implemented in \texttt R to solve~\eqref{eq:logm}.

The computer platform is Intel(R) Core~i7, with \unit[$8$]{GB} of memory and $64$ bits Windows~7.
Codes and data are available upon request from the corresponding author.

\begin{table*}[htb] 
\caption{Estimations of the $\alpha$-parameter using the samples shown in Fig.~\ref{figure:reales2}}
\label{tablaResultadosMunich}
\centering
\begin{tabular}{c*9{r}}
\toprule
 Color & size & $\widehat\alpha_{\text MV}$ &$\widehat\alpha_{\text T}$ & $\widehat\alpha_{\text{Mom12}}$ & $\widehat\alpha_{\text{LCum}}$ & time MV & time DT & time $\frac12$Mom & time LCum \\
\midrule
Magenta & $100$  &	$-1.9$ & $-2.7$    & $-1.9$   & $-1.7$  & $0.03$  & $5.85$ &$0.03$&$0.02$  \\
Green & $48$ &$-2.5$& $-2.5 $& $-2.9$& $-3.1$ & $0.00$ & $3.31$ &$ 0.00$ & $0.00$\\
Blue & $25$ & $-4.9$ &$-3.0$ &  $NA$ & $NA$ & $0.00$& $2.08$ &$0.00$ & $0.00$\\
Yellow &$90$ &$-6.2$ &$-5.1$ &$-6.6$& $-6.8$& $0.00$ &$5.16$&$0.00$ &$ 0.00$\\
 Red& $64$ & $-1.8$ &$-1.9$ &$-1.9$ & $-1.8$ &$ 0.00 $&$4.17$ &$0.00$ &$ 0.00$\\
\bottomrule
\end{tabular}
\end{table*}

\begin{table}[htb] 
\caption{Sample $p$-values of the Kolmogorov-Smirnov test with samples from the image in Fig.~\ref{figure:reales2}  }
\centering
\begin{tabular}{c*4{r}}
\toprule
& \multicolumn{4}{c}{$p$-value}\\ \cmidrule(rl){2-5}
Color & TestMV    & TestDT    & TestMom  & TestLCum \\
\midrule
Magenta &	$0.46$&$ 0.58$& $0.28$& $0.69$\\
Green & $0.37$&   $0.85$& $0.37$&$ 0.37$\\
Blue & $0.15$& $0.07$&            $  NA  $  &         $   NA $\\
Yellow &$0.63$&$0.98$&$ 0.76$&$ 0.22$\\
Red & $0.30$& $ 0.30$&    $0.21$& $0.99$\\
\bottomrule
\end{tabular}
\label{resultadosTestMunich}
\end{table}

\begin{table*}[htb] 
\centering
\caption{Estimations of the $\alpha$-parameter using the samples shown in Fig.~\ref{Cornerregiones}}
\label{resultadosCorner}
\begin{tabular}{c*{14}{r}}
\toprule
 Color & size & $\widehat\alpha_{\text MV}$ & $\widehat\alpha_{\text T}$ & $\widehat\alpha_{\text Mom12}$ & $\widehat\alpha_{\text LCum}$ & timeMV &timeDT & timeMom12 & timeLogcum12 \\
\midrule
Magenta & $15$& $-20.0 $ &$-4.1 $ &$NA $  & $NA $ &  $  0.03$ &   $ 1.95 $   &  $ 0.03$     &   $  0.03$ \\ 
Green & $42$ & $-9.2$ & $-5.0$ &$NA$ &  $ NA $ &  $  0.00$  &  $ 4.04$  &   $ 0.00$   &  $ 0.00$\\
Blue & 90 & -3.5 & -2.7 &-4.7& -14.2 &    0.02   &  4.85    &   0.00   &       0.00\\
Yellow &$156 $&$-2.2$ & $-1.8$&$ -2.6 $& $ -3.4 $ &  $ 0.01$ &  $  8.35 $  & $   0.00$     &    $ 0.00$\\
Red & $225$ & $-1.9 $ &$-1.7$ &$-2.1$  &$ -2.5 $  & $  0.02$  & $ 10.97$  &    $ 0.00$     &   $  0.00$\\
\bottomrule
\end{tabular}
\end{table*}

\begin{table}[htb] 
\centering
\caption{Sample $p$-values of the Kolmogorov-Smirnov test using the samples shown in Fig.~\ref{Cornerregiones}}
\label{pvaluesTrueNestedSamples}
\begin{tabular}{c*4{r}}
\toprule
& \multicolumn{4}{c}{$p$-value}\\ \cmidrule(rl){2-5}
 Color & TestMV    & TestDT    & TestMom  & TestLogCum\\
\midrule
Magenta&$ 0.38$& $ 0.93 $     &         $  NA  $     &        $  NA$\\
Green & $ 0.11$&$ 0.93$     &          $  NA $     &        $   NA$\\
Blue&$ 0.01 $& $0.11 $& $ 0.400$ & $ 0.11$\\
Yellow &$0.19$ &$0.31$&$ 0.460 $& $0.15$\\
Red & $0.23$ &$0.12$&$ 0.008$ &$0.02$\\
\bottomrule
\end{tabular}
\end{table}

\bibliographystyle{IEEEtran}
\bibliography{ReferencesJSTARSKernel}

\begin{IEEEbiography}[{\includegraphics[width=1in]{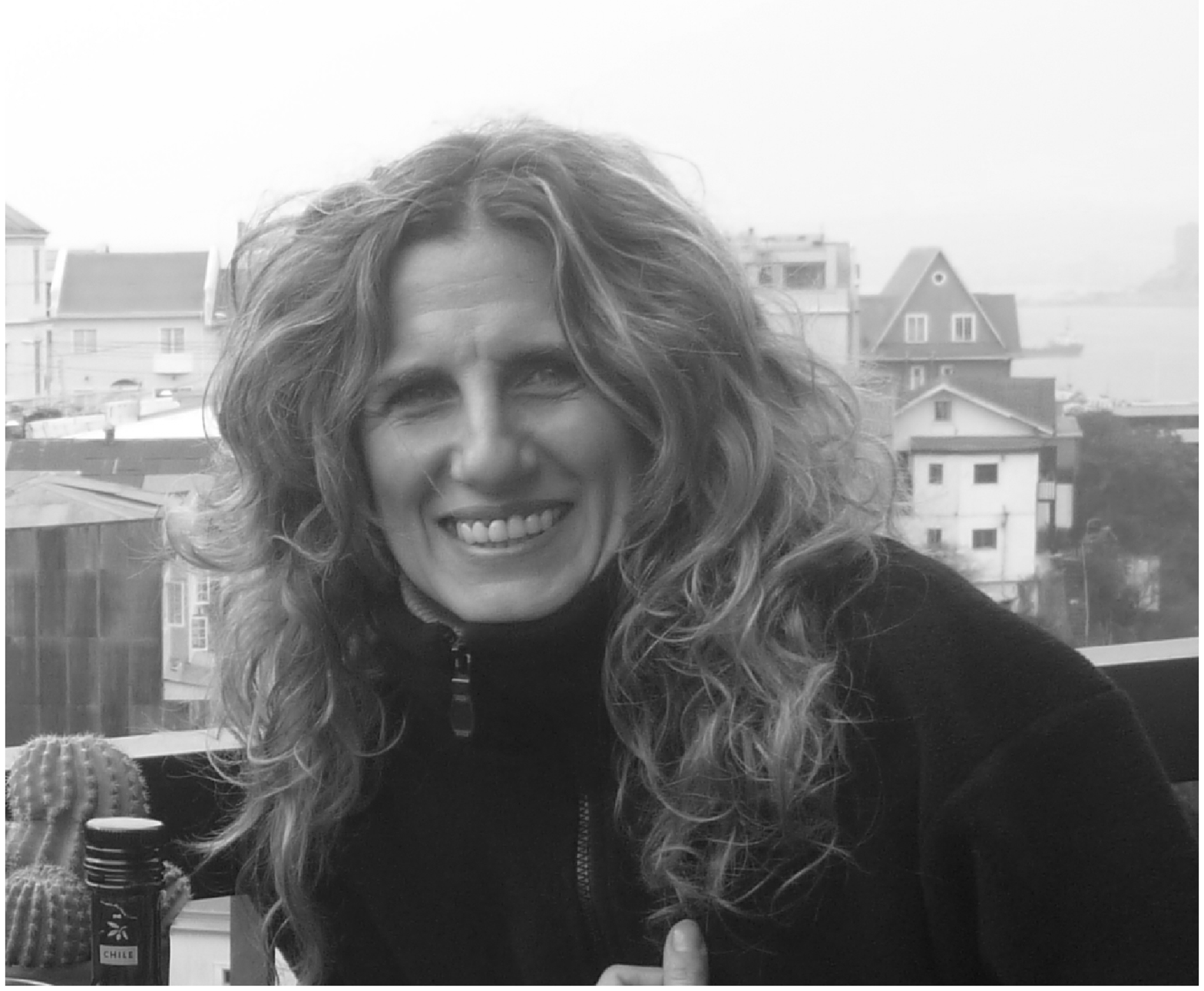}}]{Juliana Gambini} 
received the B.S. degree in Mathematics and the Ph.D. degree in Computer Science both from Universidad de Buenos Aires (UBA), Argentina.
She is currently Professor at the Instituto Tecnol\'ogico de Buenos Aires (ITBA), Buenos Aires, and Professor at Universidad Nacional de Tres de Febrero, Pcia. de Buenos Aires.
Her research interests include SAR image processing, video processing and image recognition.
\end{IEEEbiography}

\protect\vfill

\begin{IEEEbiography}[{\includegraphics[width=1in]{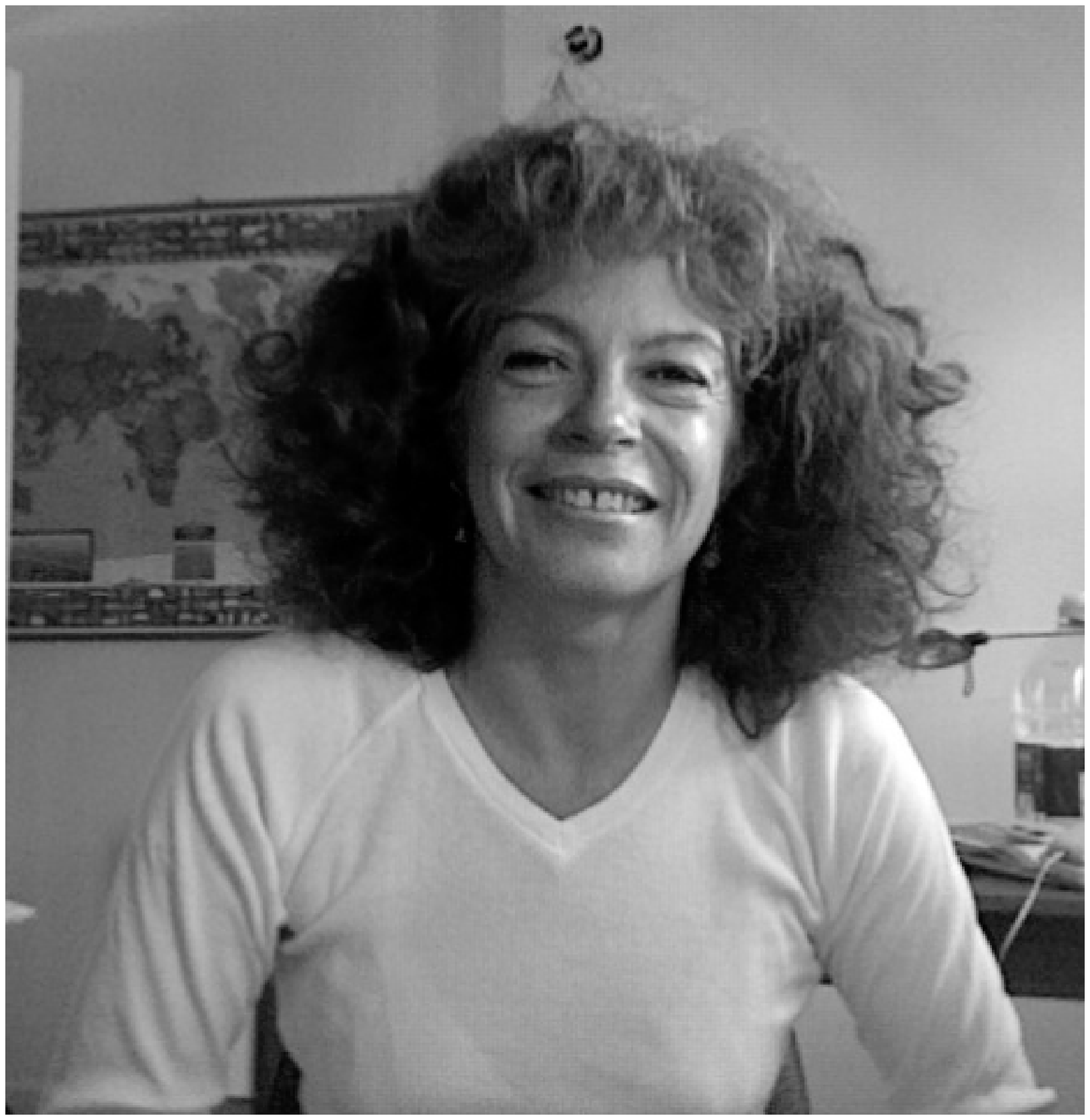}}]{Julia Cassetti} 
received the B.Sc. degree in Mathematics and a postgraduate diploma in Mathematical Statistics, both from the Universidad de Buenos Aires (UBA), Argentina. 
She is currently an Assistant Professor at the Universidad Nacional de General Sarmiento, Argentina. 
Her research interests are SAR imaging and applied statistics.
\end{IEEEbiography}

\begin{IEEEbiography}[{\includegraphics[width=1in,height=1.1in]{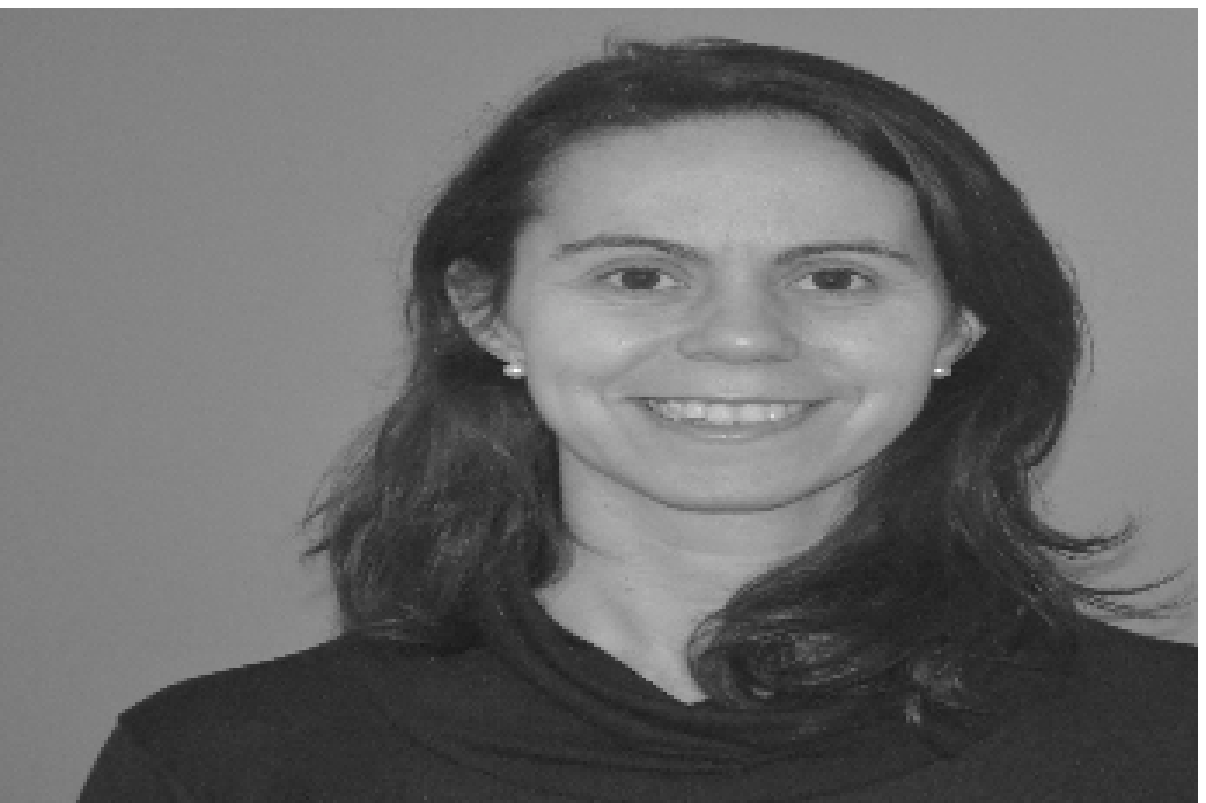}}]{Mar\'{\i}a Magdalena Lucini} received her BS and PhD degrees in mathematics from the Universidad Nacional de C\'ordoba, Argentina, in 1996 and 2002, respectively. She currently is Professor at the Universidad Nacional del Nordeste, and a researcher of CONICET, both in Argentina. Her research interests include statistical image processing and modelling, hyperspectral data, SAR and PolSAR data.
\end{IEEEbiography}

\begin{IEEEbiography}[{\includegraphics[width=1in]{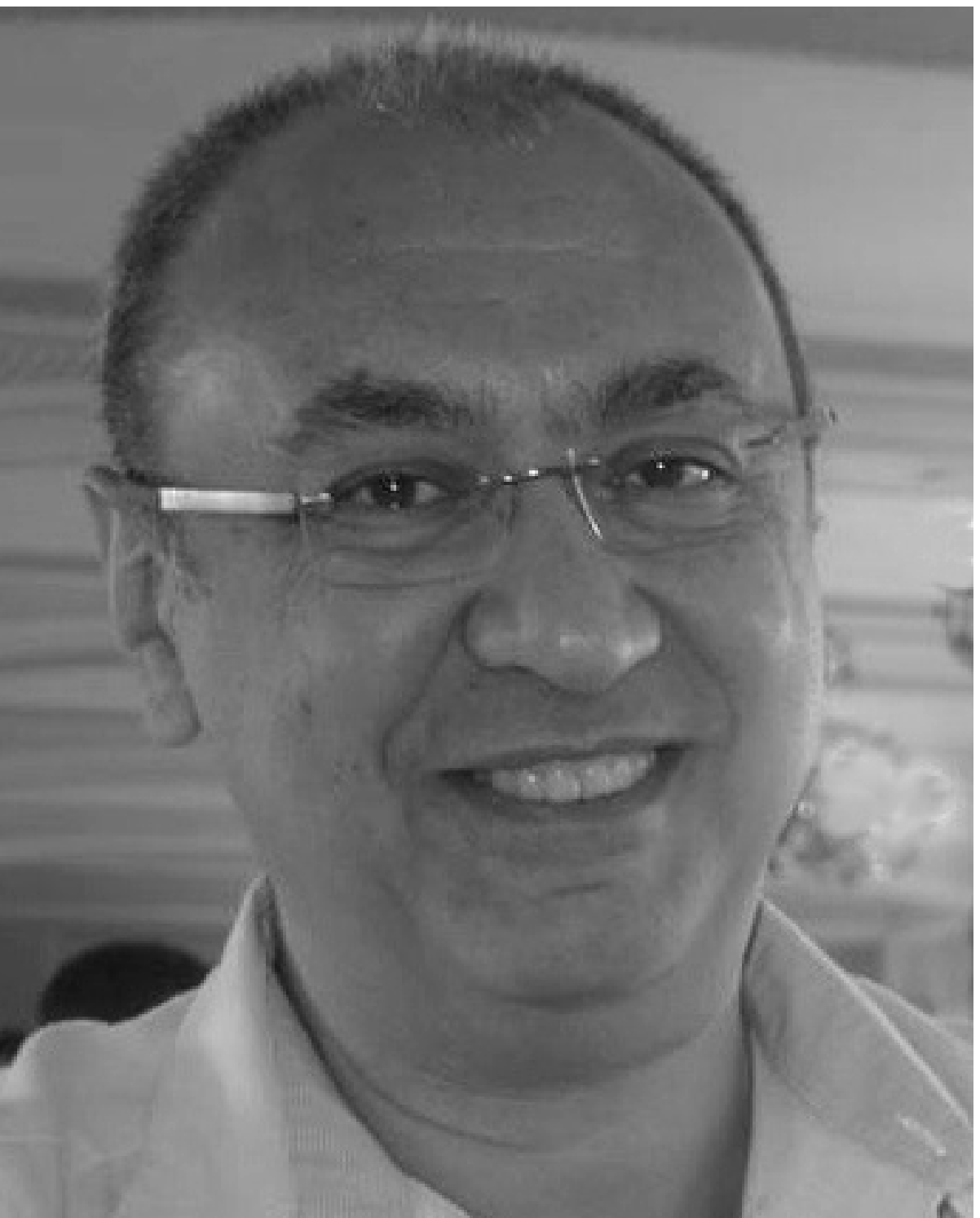}}]{Alejandro C.\ Frery} (S'92--SM'03)
received the B.Sc. degree in Electronic and Electrical Engineering from the Universidad de Mendoza, Argentina.
His M.Sc. degree was in Applied Mathematics (Statistics) from the Instituto de Matem\'atica Pura e Aplicada (IMPA, Rio de Janeiro) and his Ph.D. degree was in Applied Computing from the Instituto Nacional de Pesquisas Espaciais (INPE, S\~ao Jos\'e dos Campos, Brazil).
He is currently the leader of LaCCAN -- \textit{Laborat\'orio de Computa\c c\~ao Cient\'ifica e An\'alise Num\'erica}, 
Universidade Federal de Alagoas, Brazil.
His research interests are statistical computing and stochastic modelling.
\end{IEEEbiography}

\protect\vfill

\end{document}